\documentclass[a4paper,twoside,11pt]{amsart}
\usepackage{amsmath}
\usepackage{amsfonts}
\usepackage{amssymb}
\usepackage{amsthm}
\usepackage{newlfont}
\usepackage{graphicx}
\usepackage{amscd}

\theoremstyle{plain}
\newtheorem{Th}{Theorem}[section]
\newtheorem{Cor}[Th]{Corollary}
\newtheorem{Lem}[Th]{Lemma}
\newtheorem{Prop}[Th]{Proposition}
\theoremstyle{definition}

\newtheorem{Ex}{Example}[section]
\theoremstyle{remark}
\newtheorem*{Rem}{Remark}
\numberwithin{equation}{section}
\newcommand{\PP}{{\mathbb P}}

\newcommand{\DD}{{\mathbb D}}

\newcommand{\II}{{\mathbb I}}
\newcommand{\ZZ}{{\mathbb Z}}

\newcommand{\VV}{{\mathbb V}}

\newcommand{\bphi}{\boldsymbol{\phi}}
\newcommand{\bPhi}{\boldsymbol{\Phi}}
\newcommand{\bpsi}{\boldsymbol{\psi}}

\begin{document}

\title[Non-commutative Zamolodchikov maps, and Desargues lattices]
{Non-commutative bi-rational maps satisfying Zamolodchikov equation, and Desargues lattices}

\author[A. Doliwa]{Adam Doliwa}
\address{A. Doliwa: Faculty of Mathematics and Computer Science, University of Warmia and Mazury, 
ul.~S{\l}oneczna~54, 10-710 Olsztyn, Poland}
\email{doliwa@matman.uwm.edu.pl}
\urladdr{http://wmii.uwm.edu.pl/~doliwa/}

\author[R. M. Kashaev]{Rinat M. Kashaev}
\address{R. M. Kashaev: Section de Math\'{e}matiques, Universit\'{e} de Gen\`{e}ve, 2--4 Rue de Li\`{e}vre, Case Postale 64, 
1211 Gen\`{e}ve 4, Switzerland}
\email{Rinat.Kashaev@unige.ch}
\urladdr{http://www.unige.ch/math/folks/kashaev}

\date{}
\keywords{tetrahedron Zamolodchikov equation, functional pentagon equation, integrable discrete geometry; non-commutative rational functions; non-Abelian Hirota--Miwa system; quantum maps; Poisson maps}
\subjclass[2010]{Primary 37K60; Secondary 39A14, 51A20, 16T25, 81R12}

\begin{abstract}
We present new solutions of the functional Zamolodchikov tetrahedron equation in terms of birational maps in totally non-commutative variables. All the maps originate from  Desargues lattices, which provide geometric realization of solutions to the non-Abelian Hirota--Miwa system. The first map is derived using the original Hirota's gauge for the corresponding linear problem, and the second one from its affine (non-homogeneous) description. We provide also an interpretation of the maps within the local Yang--Baxter equation approach.
 We exploit decomposition of the second map into two simpler maps which, as we show, satisfy the pentagonal condition. We provide also geometric meaning of the matching ten-term condition between the pentagonal maps. 
The generic description of Desargues lattices in homogeneous coordinates allows to define another solution of the Zamolodchikov equation, but with functional parameter which should be adjusted in a particular way. Its ultra-local reduction produces a birational quantum map (with two central parameters) with Zamolodchikov property, which preserves Weyl commutation relations.  In the classical limit our construction gives the corresponding Poisson map satisfying the Zamolodchikov condition. 
\end{abstract}
\maketitle

\section{Introduction}
Let $\mathcal{X}$ be any set, a map $R:\mathcal{X}^3\to \mathcal{X}^3$ is called 
\emph{tetrahedral or Zamolodchikov map} when in $\mathcal{X}^6$ the relation 
\begin{equation} \label{eq:f-tetrahedron}
R_{123} \circ R_{145} \circ R_{246}\circ R_{356} = 
R_{356} \circ R_{246} \circ R_{145} \circ R_{123}, 
\end{equation}
holds. Here $R_{ijk}$ acts as $R$ on the $i$-th, $j$-th and $k$-th factors and as identity on others.
\begin{figure}[h!]
	\begin{center}
		\includegraphics[width=6cm]{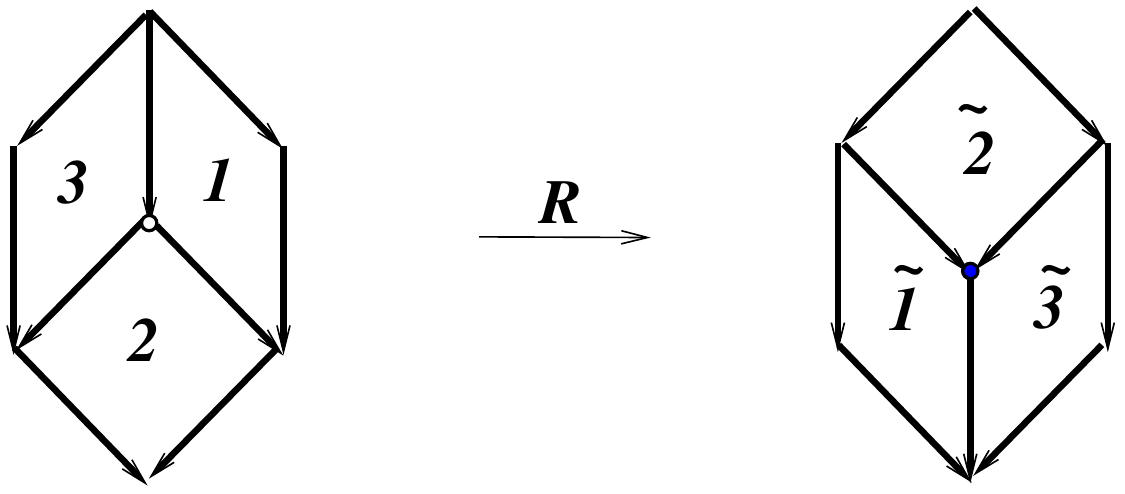}
	\end{center}
	\caption{Graphical representation of a map $R:\mathcal{X}^3\to \mathcal{X}^3$ in quadrilateral formalism}
	\label{fig:Zamolodchikov-map}
\end{figure}
By representing graphically the map as a flip between faces of a cube, as given in Figure~\ref{fig:Zamolodchikov-map}, the Zamolodchikov condition can be visualized in terms of flips between 2-dimensional faces of the 4-dimensional cube (tesseract), see Figure~\ref{fig:Zamolodchikov}. This graphical representation closely resembles the geometric realization of 4-dimensional consistency of the discrete Darboux equations in terms of lattices of planar quadrilaterals~\cite{MQL}.

The operator version of the equation~\eqref{eq:f-tetrahedron} was proposed by Zamolodchikov~\cite{Zamolodchikov} as a multidimensional version of the quantum Yang--Baxter equation, which plays central role in two-dimensional solvable models of statistical and quantum physics~\cite{Baxter,QISM,Jimbo}. A Yang--Baxter map (or set-theoretical solution to quantum Yang--Baxter equation \cite{Drinfeld}) $Y \colon \mathcal{X}^2\to \mathcal{X}^2$ satisfies functional equation
\begin{equation}
\label{eq:f-YB}
Y_{12}\circ Y_{13} \circ Y_{23} = Y_{23}\circ Y_{13} \circ Y_{12} \qquad \text{in} \qquad
\mathcal{X}^3 .
\end{equation}
After first solutions of the quantum Zamolodchikov (or tetrahedron) equation have been constructed~\cite{Baxter-Z,BaStro,Korepanov}, the equation was described in the language of category theory~\cite{KapranovVoevodsky} and has found applications to topology of knotted surfaces~\cite{CarterSaito}. 

Those solutions of the Zamolodchikov equation, which can be constructed form solutions to the quantum Yang--Baxter equation are considered as trivial (independently from their complexity). Solutions to the quantum Zamolodchikov equation related to discrete Darboux equations, which have been found in~\cite{BaMaSe,BaSe,Sergeev-q3w}, are not of this type. It turns out that such genuine
solutions of the quantum Zamolodchikov equation give rise to solutions of the quantum Yang--Baxter equation in the reduction process and allow for better understanding of algebraic structures behind two-dimensional integrable models~\cite{Kashaev-Volkov,BaMaSe,KunibaSergeev}. 
As a rule, given $3$-dimensional integrable $4$-dimensionally consistent system, it allows to construct the corresponding solution of the functional Zamolodchikov equation~\eqref{eq:f-tetrahedron}, see for example~\cite{Kashaev-Korepanov-Sergeev,Kashaev-Sergeev,Kashaev-LMP,KNPT}.  
\begin{figure}
	\begin{center}
		\includegraphics[width=14cm]{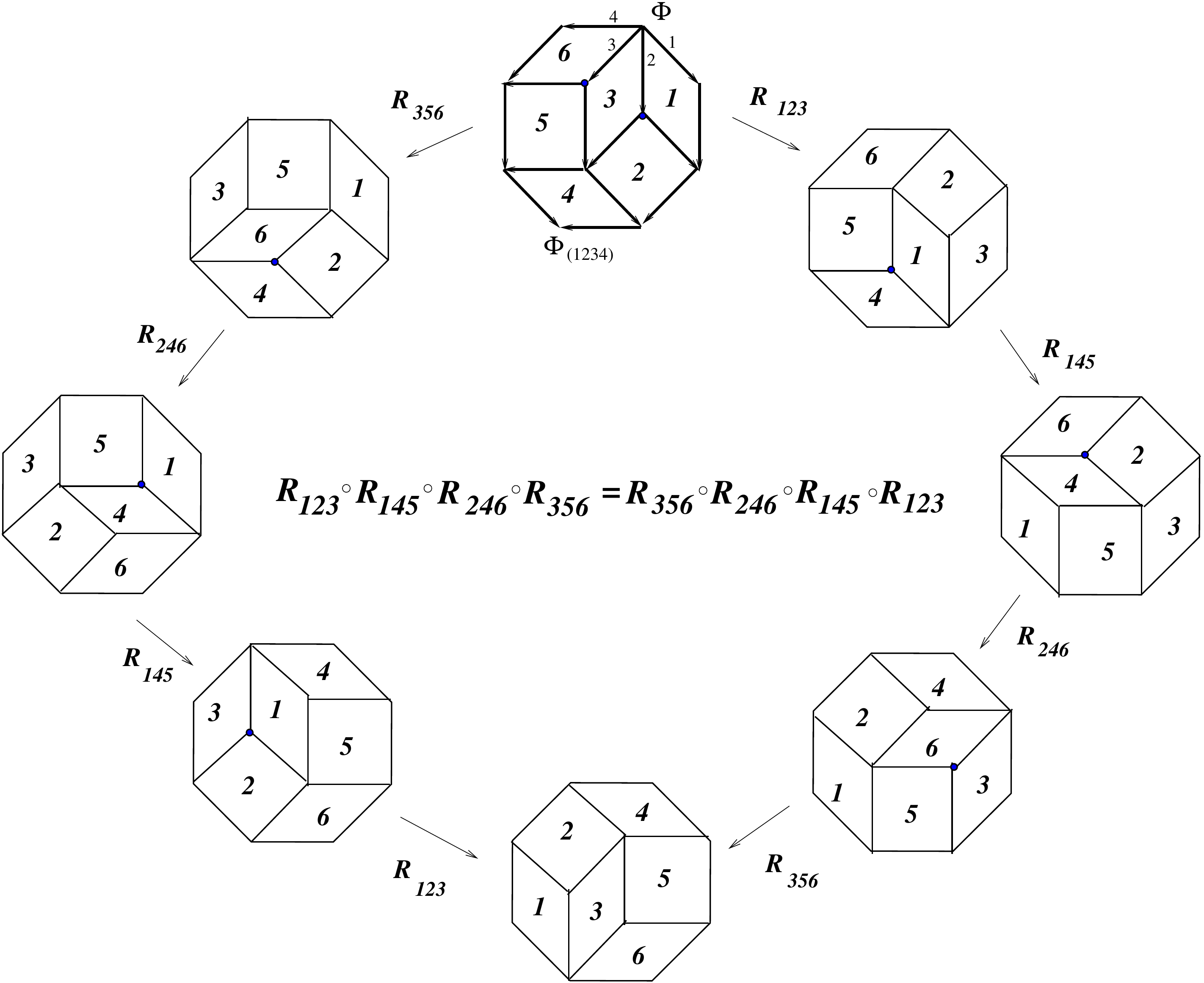}
	\end{center}
	\caption{The graphical representation of Zamolodchikov's equation}
	\label{fig:Zamolodchikov}
\end{figure}

Non-commutative extensions of integrable systems are of growing interest in mathematical physics~\cite{FWN-Capel,Kupershmidt,BobSur-nc,Nimmo-NCKP,Dol-GD,DoliwaNoumi}. They may be considered as a useful platform to more thorough understanding of integrable quantum or statistical mechanical lattice systems. 
As it was noticed in~\cite{Dol-Des} the projective incidence geometric approach to integrable discrete systems points out their non-commutative versions.  
In the present paper we consider Zamolodchikov maps in non-commuting variables which are related to the non-Abelian Hirota--Miwa system~\cite{Nimmo-NCKP} and its projective geometric description~\cite{Dol-Des}. We exploit decomposition of the maps into two pentagonal maps introduced in~\cite{DoliwaSergeev-pentagon}, and study the matching condition between them~\cite{Kashaev-Sergeev}. 
The geometric approach turns out to be very helpful, because the most general description of the maps in terms of homogeneous coordinates of the projective space allows to construct the corresponding quantum map satisfying the Zamolodchikov condition. The map preserves the Weyl commutation relations, and in the classical limit it gives the Zamolodchikov map which preserves the corresponding Poisson brackets.

The structure of paper is as follows. In Section~\ref{sec:Des-Hir-Zam} we first recall the geometric description of the non-Abelian Hirota--Miwa system (non-commutative discrete KP system) in terms of the Desargues maps. This allows to present the corresponding non-commutative birational map --- the Hirota map in the original gauge --- which satisfies the Zamolodchikov equation. By changing the gauge to affine space description, i.e. to non-homogeneous coordinates, we obtain another non-commutative birational map with Zamolodchikov property --- the affine Hirota map. The next Section~\ref{sec:pentagonal-decomposition} is devoted to decomposition of the affine Hirota map into affine versions of the non-commutative normalization and Veblen maps. These birational maps, introduced in a different gauge in~\cite{DoliwaSergeev-pentagon} satisfy the pentagonal condition. Following ideas of~\cite{Kashaev-Sergeev} we prove that these two maps are paired by the ten-term relation. In particular we present geometric meaning of the relation in terms of the star configuration. In Section~\ref{sec:NP-hom}, following mainly~\cite{DoliwaSergeev-pentagon} but changing slightly some formulas, we present the normalization and Veblen maps using homogeneous coordinates of the projective space. We also present ultra-local/quantum reductions of the maps and their classical limits to corresponding Poisson maps satisfying pentagonal condition. Finally, in Section~\ref{sec:Zam-h} we combine the homogeneous versions of the pentagonal maps into the corresponding homogeneous Hirota map. We present also the geometric meaning of the corresponding local Yang--Baxter equation. We provide direct proof that these maps are matched by the homogeneous version of the ten-term relation and we give its quantum/Poisson version as well. This provides the corresponding quantum/Poisson Hirota map with the Zamolodchikov property.

Throughout the paper we will work with division rings of (non-commutative) \emph{rational functions} in a finite number of (non-commuting) variables. This approach is intuitively accessible, see however \cite{Cohn} for formal definitions.

Results of the paper were presented by A.~D. on conferences \emph{Symmetries and Integrability of Difference Equations~12} (Sainte-Ad\`{e}le, Qu\'{e}bec, Canada, 3-9 July, 2016), and \emph{Integrable Systems and Quantum Symmetries XXV} (Prague, 6-10 June, 2017).

\section{Desargues maps in the Hirota and affine gauges, and the corresponding solutions to the Zamolodchikov equation} \label{sec:Des-Hir-Zam}

\subsection{Desargues maps and the non-Abelian Hirota--Miwa system}
Consider the following linear problem \cite{DJM-II,Nimmo-NCKP}
\begin{equation} \label{eq:lin-dKP}
\bphi_{(i)} - \bphi_{(j)} =  \bphi U_{ij},  \qquad i \ne j \leq N,
\end{equation}
where $U_{ij}\colon\ZZ^N \to \DD$ are functions defined on $N$-dimensional integer lattice with values in a division ring $\DD$, and the wave function $\bphi\colon \ZZ^N \to \VV(\DD)$ takes values in a right vector space over $\DD$. Here and in all the paper the subscripts in brackets denote shifts in corresponding discrete variables, i.e., $f_{(i)}(n_1, \dots , n_i, \dots, n_N) = f(n_1,\dots , n_i + 1, \dots ,n_N)$.

The compatibility conditions of \eqref{eq:lin-dKP} consist of equations 
\begin{equation} \label{eq:alg-comp-U}
U_{ij} + U_{ji} = 0, \qquad  U_{ij} + U_{jk} + U_{ki} = 0, \qquad 
U_{kj}U_{ki(j)} = U_{ki} U_{kj(i)},
\end{equation}
called in \cite{Nimmo-NCKP} non-Abelian Hirota--Miwa system of equations. Indeed, for commutative $\DD$ it is possible to introduce the $\tau$-function such that
\begin{equation} \label{eq:U-tau}
U_{ij} = \frac{\tau \tau_{(ij)}}{\tau_{(i)} \tau_{(j)}}, \qquad i< j.
\end{equation}
The remaining part of the system
\eqref{eq:alg-comp-U} reduces to Hirota's discrete KP equation \cite{Hirota}
\begin{equation} \label{eq:H-M}
\tau_{(i)}\tau_{(jk)} - \tau_{(j)}\tau_{(ik)} + \tau_{(k)}\tau_{(ij)} =0,
\qquad 1\leq i< j <k \leq N,
\end{equation}
whose pivotal role in the whole KP hierarchy was discovered by Miwa~\cite{Miwa}. We remark that on the level the $\tau$-function formulation the quantum analog of equation~\eqref{eq:H-M}, different from that presented in Section~\ref{sec:Hir-q-P} of the present paper, is given in~\cite{Kashaev-qH,Kashaev-Reshetikhin}.

When values of the linear function $\bphi$ are interpreted as homogeneous coordinates of a map $[\bphi ]\colon \ZZ^N \to \PP(\VV)$ from the integer lattice into the projectivisation of $\VV$ then the linear problem \eqref{eq:lin-dKP} simply states that the points $[ \bphi ]$,  $[ \bphi_{(i)} ]$ and  $[ \bphi_{(j)} ]$ are collinear, see Figure~\ref{fig:Desargues-map}.  
\begin{figure}[h!]
	\begin{center}
		\includegraphics[width=10cm]{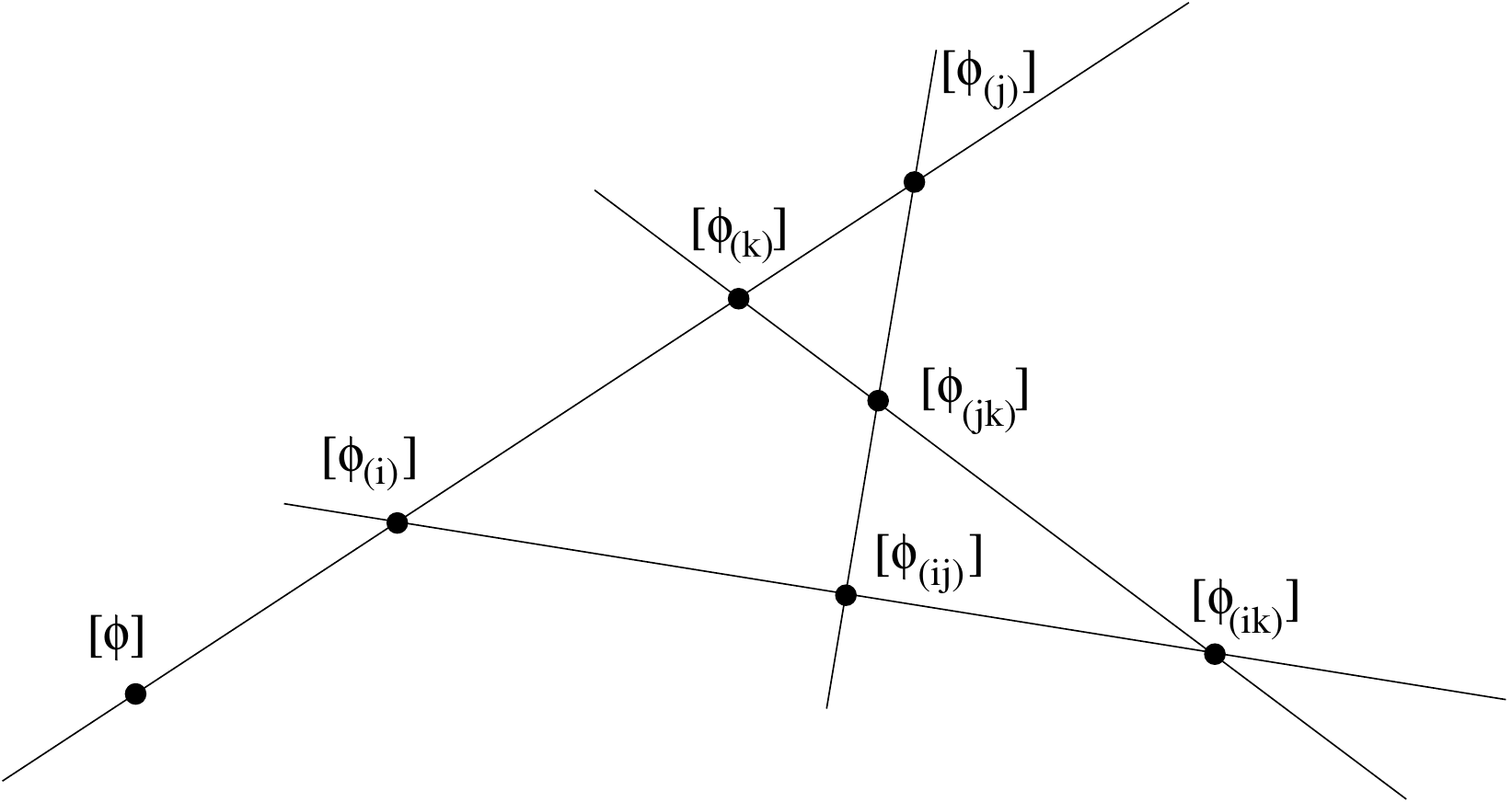}
	\end{center}
	\caption{Desargues map condition}
	\label{fig:Desargues-map}
\end{figure}
Such geometric maps were called Desargues maps in \cite{Dol-Des}, because their four-dimensional consistency, which is crucial for integrability of~\eqref{eq:alg-comp-U},  is encoded in the Desargues theorem of projective geometry~\cite{BeukenhoutCameron-H}.

\subsection{The non-commutative Hirota map} \label{sec:Hir-Zam-H}
Let us present the corresponding solution to the functional Zamolodchikov equation together with its linear problem. We will keep both the geometric meaning and the form of the linear problem.
The following non-commutative Hirota map $H \colon (y_1,y_2,y_3) \dashrightarrow (\tilde{y}_1,\tilde{y}_2,\tilde{y}_3)$ can be extracted from compatibility of the linear system
\begin{align} \nonumber
\bphi_{(1)} - \bphi_{(2)} & = \bphi y_1 , & \bphi_{(13)} - \bphi_{(23)} & = \bphi_{(3)} \tilde{y}_1, \\ \label{eq:lin-H}
\bphi_{(12)} - \bphi_{(23)} & = \bphi_{(2)} y_2 , & \bphi_{(1)} - \bphi_{(3)} & = \bphi \tilde{y}_2 , \\ \nonumber
\bphi_{(2)} - \bphi_{(3)} & = \bphi y_3 , & \bphi_{(12)} - \bphi_{(13)} & = \bphi_{(1)} \tilde{y}_3 ,
\end{align}
and reads
\begin{align} 
\nonumber \tilde{y}_1 & = (y_1 + y_3)^{-1} y_1 y_2 , \\
\label{eq:Hirota-map}
\tilde{y}_2 & = y_1 + y_3 , \\
\nonumber \tilde{y}_3 & = (y_1 + y_3)^{-1} y_3 y_2 .
\end{align}

The  Hirota map is direct consequence of the non-Abelian Hirota--Miwa system, and the identification of the fields is given by
\begin{align*}
y_1 & = U_{12},   & y_2 &= U_{13(2)} ,  & y_3 &= U_{23} ,\\
\tilde{y}_1 & = U_{12(3)} , & \tilde{y}_2 &= U_{13} ,  & \tilde{y}_3 &= U_{23(1)} .
\end{align*}
By simple calculation on can check the following property of the map $H$.
\begin{Prop}
	The Hirota map $H$ given by \eqref{eq:Hirota-map} is birational with inverse being its opposite map $H^\mathrm{op}$, i.e. with all multiplications taken in opposite order
	\begin{align} 
	\nonumber y_1 & = \tilde{y}_2 \tilde{y}_1 (\tilde{y}_1 + \tilde{y}_3)^{-1} , \\
	\label{eq:Hirota-map-inv}
	y_2 & = \tilde{y}_1 + \tilde{y}_3 , \\
	\nonumber y_3 & = \tilde{y}_2 \tilde{y}_3 (\tilde{y}_1 + \tilde{y}_3)^{-1}  .
	\end{align}
\end{Prop}
The commutative version of the Hirota map~\eqref{eq:Hirota-map} was constructed in~\cite{Sergeev-LMP} as a limiting case of a map defined by a local Yang--Baxter equation. Let us provide such a direct (without taking any limits) interpretation valid also in the non-commutative case. We postpone geometric explanation of the fact to Section~\ref{sec:Hir-h}.
\begin{Prop}
Define $2\times 2$ matrix $L^H(y)=\left( \begin{array}{cc} y & 1 \\ 1 & 0 \end{array} \right)$ and its $3\times 3$ extensions by
	\begin{equation*}
	L^H_{12}(y) = \left( \begin{array}{ccc} y & 1 & 0 \\ 1 & 0 & 0 \\ 0 & 0 & 1\end{array} \right) , \quad
	L^H_{13}(y) = \left( \begin{array}{ccc} y & 0 & 1 \\ 0 & 1 & 0 \\ 1 & 0 & 0\end{array} \right) , \quad
	L^H_{23}(y) = \left( \begin{array}{ccc} 1 & 0 & 0 \\ 0 & y & 1 \\ 0 & 1 & 0\end{array} \right) .
	\end{equation*}
	The corresponding local Yang--Baxter equation 
	\begin{equation} \label{eq:lYB-H}
	L_{12}^H (y_1) L_{13}^H (y_2) L^H_{23}(y_3) = 
	L_{23}^H(\tilde{y}_3) L_{13}^H(\tilde{y}_2) L_{12}^H(\tilde{y}_1),
	\end{equation}
	reads explicitly as follows
	\begin{equation*}
	\left( \begin{array}{ccc} y_1 y_2 & y_3 + y_1 & 1 \\ y_2 & 1 & 0 \\ 
	1 & 0 & 0\end{array} \right) = 
	\left( \begin{array}{ccc} \tilde{y}_2 \tilde{y}_1 & \tilde{y}_2 & 1 \\ 
	 \tilde{y}_1 + \tilde{y}_3 & 1 & 0 \\ 
	1 & 0 & 0\end{array} \right),
	\end{equation*} 
	and is equivalent to the transformation formulas~\eqref{eq:Hirota-map}.
\end{Prop}
\begin{Rem}
	In \cite{Kashaev-qH} one can find another matrix factorization description of the (commutative) Hirota map in terms of the Heisenberg group which works also in the non-commutative case
	\begin{equation*}
 \left( \begin{array}{ccc} 1 & y_1 & 0 \\ 0 & 1 & 0 \\ 0 & 0 & 1\end{array} \right)
 \left( \begin{array}{ccc} 1 & 0 & 0 \\ 0 & 1 & y_2 \\ 0 & 0 & 1\end{array} \right)
 \left( \begin{array}{ccc} 1 & y_3 & 0 \\ 0 & 1 & 0 \\ 0 & 0 & 1\end{array} \right)=
 \left( \begin{array}{ccc} 1 & 0 & 0 \\ 0 & 1 & \tilde{y}_3 \\ 0 & 0 & 1\end{array} \right)
 \left( \begin{array}{ccc} 1 & \tilde{y}_2 & 0 \\ 0 & 1 & 0 \\ 0 & 0 & 1\end{array} \right)
  \left( \begin{array}{ccc} 1 & 0 & 0 \\ 0 & 1 & \tilde{y}_1 \\ 0 & 0 & 1\end{array} \right).
	\end{equation*}
\end{Rem}

Let us state the final result of this Section.
\begin{Prop} \label{prop:H-Z}
	The non-commutative Hirota map satisfies functional Zamolodchikov equation
	\begin{equation}
H_{123} \circ H_{145} \circ H_{246}\circ  H_{356} = H_{356} \circ H_{246} \circ H_{145} \circ H_{123}.
	\end{equation}
\end{Prop}
\begin{proof}
	By direct calculation. We provide only the resulting expressions
	\begin{equation} \label{eq:H-Z-yyy}
	(y_1^\prime , y_2^\prime ,y_3^\prime , y_4^\prime , y_5^\prime , y_6^\prime ) = 
	\begin{cases}
(H_{123} \circ H_{145} \circ H_{246}\circ  H_{356})(y_1,y_2,y_3,y_4,y_5,y_6) , \\ 
(H_{356} \circ H_{246} \circ H_{145} \circ H_{123})	(y_1,y_2,y_3,y_4,y_5,y_6) ,
\end{cases}
	\end{equation}
	where
	\begin{equation}
	\begin{aligned}
	y_1^\prime & = \left( y_1 y_2 + y_1 y_5 + y_3 y_5 \right)^{-1} y_1 y_2 y_4 ,\\
	y_2^\prime & = \left( y_1 + y_3 + y_6 \right)^{-1} \left( y_1 y_2 + y_1 y_5 + y_3 y_5\right) ,\\
	y_3^\prime & = y_4 - \left( y_1 y_2 + y_1 y_5 + y_3 y_5 \right)^{-1} y_1 y_2 y_4 
	- \left( y_3 y_2 + y_6 y_2 + y_6 y_5\right)^{-1} y_6 y_5 y_4 ,\\
	y_4^\prime & = y_1 + y_3 + y_6 , \\
	y_5^\prime & = \left( y_1 + y_3 + y_6 \right)^{-1} \left( y_3 y_2 + y_6 y_2 + y_6 y_5\right) , \\
	y_6^\prime & = \left( y_3 y_2 + y_6 y_2 + y_6 y_5\right)^{-1} y_6 y_5 y_4.
	\end{aligned}
	\label{eq:H-Z-yy}
	\end{equation}
\end{proof} 
\begin{Rem}
	Using the above-mentioned identification of the fields entering the  Hirota map and the non-Abelian Hirota--Miwa system we have
	\begin{align*}
	y_1 & = U_{12} ,  & y_2 &= U_{13(2)} ,  & y_3 &= U_{23} , &	y_4 & = U_{14(23)}  , & y_5 &= U_{24(3)}  , & y_6 &= U_{34 ,}\\
	y_1^\prime & = U_{12(34)} ,  & y_2^\prime &= U_{13(4)} ,  & y_3^\prime &= U_{23(14)}, &	y_4^\prime & = U_{14} ,  & y_5^\prime &= U_{24(1)}  , & y_6^\prime &= U_{34(12)}.
	\end{align*}
\end{Rem}
\begin{Cor}
	The inverse of the Hirota map satisfies the Zamolodchikov equation as well.
\end{Cor}
Let us give more conceptual proof of Proposition~\ref{prop:H-Z} using realization of the Hirota map in terms of the local Yang--Baxter equation \eqref{eq:lYB-H}. For $1\leq i < j \leq 4$ by $L^H_{ij}(y)$ denote the corresponding $4\times 4$ extension of the $2\times 2$ matrix $L^H(y)$. Then, given six non-commuting variables $(y_1, \dots , y_6)$, consider the equation  
\begin{equation*} 
L_{12}^H (y_1) L_{13}^H (y_2) L^H_{23}(y_3) L_{14}^H (y_4) L_{24}^H (y_5) L^H_{34}(y_6) = 
 L^H_{34}(y_6^\prime)  L_{24}^H (y_5^\prime) L_{14}^H (y_4^\prime) L_{23}^H(y^\prime_3) L_{13}^H(y^\prime_2) L_{12}^H(y^\prime_1),
\end{equation*}
which reads explicitly as follows
\begin{equation*}
\left( \begin{array}{cccc} y_1 y_2 y_4 & y_1 y_2 + y_3 y_5 + y_1y_5 & y_1 + y_3 + y_6 & 1 \\  y_2 y_4 & y_2 + y_5  & 1 & 0 \\ 
y_4 & 1 & 0 & 0 \\ 1 & 0 & 0 & 0\end{array} \right) = 
\left( \begin{array}{cccc}  y_4^\prime   y_2^\prime  y_1^\prime  &  y_4^\prime y_2^\prime &  y_4^\prime & 1 \\ y_2^\prime y_1^\prime + y_5^\prime y_1^\prime + y_5^\prime y_3^\prime  & y_2^\prime + y_5^\prime  & 1 & 0 \\ 
y_1^\prime + y_3^\prime + y_6^\prime  & 1 & 0 & 0 \\ 1 & 0 & 0 & 0\end{array} \right).
\end{equation*} 
The above system of equations has unique solution $(y_1^\prime, \dots , y_6^\prime)$ given by the transformation formulas~\eqref{eq:H-Z-yy}. The refactorization can be made in four steps using the local Yang--Baxter equation \eqref{eq:lYB-H} transitions and commutativity of the matrices acting in different spaces. To make formulas more readable define
\begin{align*}
(\tilde{y}_1^{(1)}, \dots , \tilde{y}_6^{(1)}) & = H_{123}(y_1,\dots , y_6) , & 
(\tilde{y}_1^{(2)}, \dots , \tilde{y}_6^{(2)}) & = H_{145}(\tilde{y}_1^{(1)}, \dots , \tilde{y}_6^{(1)}) \\
(\tilde{y}_1^{(3)}, \dots , \tilde{y}_6^{(3)}) & = H_{246}(\tilde{y}_1^{(2)}, \dots , \tilde{y}_6^{(2)}) ,  & (\tilde{y}_1^{(4)}, \dots , \tilde{y}_6^{(4)}) & = H_{356}(\tilde{y}_1^{(3)}, \dots , \tilde{y}_6^{(3)}) .
\end{align*}
The upper equality of \eqref{eq:H-Z-yyy} can be obtained by the following sequence of transitions
\begin{gather*}
\underline{ L_{12}^H (y_1) L_{13}^H (y_2) L^H_{23}(y_3)} L_{14}^H (y_4) L_{24}^H (y_5) L^H_{34}(y_6) = \\ =
L_{23}^H (\tilde{y}_3^{(1)}) L_{13}^H (\tilde{y}_2^{(1)}) \underline{ L^H_{12}(\tilde{y}_1^{(1)}) L_{14}^H (\tilde{y}_4^{(1)}) L_{24}^H (\tilde{y}_5^{(1)})} L^H_{34}(\tilde{y}_6^{(1)}) =\\ =
L_{23}^H (\tilde{y}_3^{(2)}) 
L_{24}^H (\tilde{y}_5^{(2)}) \underline{ L_{13}^H (\tilde{y}_2^{(2)})  L_{14}^H (\tilde{y}_4^{(2)})  L^H_{34}(\tilde{y}_6^{(2)}) } L^H_{12}(\tilde{y}_1^{(2)})  = \\ =
\underline{ L_{23}^H (\tilde{y}_3^{(3)}) 
L_{24}^H (\tilde{y}_5^{(3)}) 
L^H_{34}(\tilde{y}_6^{(3)})}
L_{14}^H (\tilde{y}_4^{(3)})   
 L_{13}^H (\tilde{y}_2^{(3)})   L^H_{12}(\tilde{y}_1^{(3)})  = \\ =
 	L^H_{34}(\tilde{y}_6^{(4)}) 	L_{24}^H (\tilde{y}_5^{(4)}) 
  L_{14}^H (\tilde{y}_4^{(4)})   L_{23}^H (\tilde{y}_3^{(4)})  
 L_{13}^H (\tilde{y}_2^{(4)})   L^H_{12}(\tilde{y}_1^{(4)}) .
\end{gather*}
To obtain the lower equality of \eqref{eq:H-Z-yyy} first commute $L^H_{23}(y_3) L_{14}^H (y_4) = L_{14}^H (y_4)L^H_{23}(y_3) $ and start making transitions from the other side.

\subsection{Desargues maps in the affine gauge}
Geometry of the Desargues maps does not change if we multiply the wave function by an arbitrary gauge function. Various gauge-equivalent forms of the corresponding linear problem were discussed in~\cite{Dol-Des}. In particular, in the affine gauge, where the collinearity condition is expressed as proportionality of vectors along the line, the linear problem gauge-equivalent to \eqref{eq:lin-dKP} reads
\begin{equation} \label{eq:lin-mdKP}
\bPhi_{(j)} - \bPhi  = (\bPhi_{(i)} - \bPhi ) B_{ij},  \qquad i \ne j \leq N.
\end{equation}
It is known that the compatibility condition of the above linear system leads to the 
non-commutative discrete mKP system~\cite{FWN-Capel}. Under identification of the fields given by
\begin{align*}
x_1 & = 1 - B_{12},   & x_2 &= 1 - B_{13(2)} ,  & x_3 &= 1 - B_{23} ,\\
\tilde{x}_1 & = 1 - B_{12(3)} , & \tilde{x}_2 &= 1 - B_{13} ,  & \tilde{x}_3 &= 1 - B_{23(1)} .
\end{align*}
we have then the folowing linear problem gauge-equivalent to
\eqref{eq:lin-H}
\begin{align} \nonumber
\bPhi_{(2)} -  \bPhi_{(1)}& = (\bPhi -  \bPhi_{(1)}) x_1 , & 
\bPhi_{(23)} - \bPhi_{(13)} & = (\bPhi_{(3)} -\bPhi_{(13)}) \tilde{x}_1 , \\
\bPhi_{(23)} - \bPhi_{(12)} & = (\bPhi_{(2)} - \bPhi_{(12)} ) x_2 , & 
\bPhi_{(3)} - \bPhi_{(1)} & = (\bPhi - \bPhi_{(1)} ) \tilde{x}_2, 
\label{eq:lin-aff} \\
\bPhi_{(3)} - \bPhi_{(2)} & = (\bPhi -\bPhi_{(2)} )x_3 , &
\bPhi_{(13)} - \bPhi_{(12)}  & = (\bPhi_{(1)} - \bPhi_{(12)}  ) \tilde{x}_3 .
\nonumber
\end{align}
The following non-commutative affine Hirota  map  $F(x_1,x_2,x_3) = (\tilde{x}_1,\tilde{x}_2,\tilde{x}_3)$, which  can be extracted from compatibility of the above  linear system is given by
\begin{align} \nonumber
\tilde{x}_1 & = \left[ x_3 + x_1(1-x_3)\right]^{-1} x_1 x_2, \\ 
\label{eq:F-map}
\tilde{x}_2 & = x_3 + x_1(1-x_3), \\
\nonumber
\tilde{x}_3 & = 1 + (x_2-1) \left[ (1-x_1) x_3 + x_1(1-x_2)\right]^{-1} \left( x_3 + x_1 (1-x_3) \right).
\end{align}
\begin{Prop}
	The affine Hirota map $F \colon (\tilde{x}_1,\tilde{x}_2,\tilde{x}_3) \dashrightarrow 
	(x_1,x_2,x_3) $ is birational with its inverse being the opposite map $F^\mathrm{op}$ 
	\begin{align*}
	x_1 & = \tilde{x}_2 \tilde{x}_1 \left[ \tilde{x}_3 + (1-\tilde{x}_3)\tilde{x}_1 \right]^{-1}, \\
	x_2 & = \tilde{x}_3 + (1-\tilde{x}_3)\tilde{x}_1, \\
	x_3 & = 1 + \left( \tilde{x}_3 + (1-\tilde{x}_3)\tilde{x}_1 \right) \left[ \tilde{x}_3(1-\tilde{x}_1) + (1-\tilde{x}_2) 
	\tilde{x}_1 \right]^{-1} (\tilde{x}_2-1),
	\end{align*}
	and satisfies the functional Zamolodchikov equation.
\end{Prop}
In~\cite{Sergeev-LMP} one can find also the local Yang--Baxter equation derivation of the affine Hirota map~\eqref{eq:F-map} for commuting variables. It turns out that the ansatz presented there is  valid also in the non-commutative case. Again we only state the result, which can be easily checked, postponing its geometric explanation to Section~\ref{sec:Hir-h}.
\begin{Prop}
	Define $2\times 2$ matrix $L^A(x)=\left( \begin{array}{cc} x & 1 \\ 1 - x& 0 \end{array} \right)$ and its $3\times 3$ extensions by
	\begin{equation*}
	L^A_{12}(x) = \left( \begin{array}{ccc} x & 1 & 0 \\ 1-x & 0 & 0 \\ 0 & 0 & 1\end{array} \right) , \quad
	L^A_{13}(x) = \left( \begin{array}{ccc} x & 0 & 1 \\ 0 & 1 & 0 \\ 1-x & 0 & 0\end{array} \right) , \quad
	L^A_{23}(x) = \left( \begin{array}{ccc} 1 & 0 & 0 \\ 0 & x & 1 \\ 0 & 1-x & 0\end{array} \right) .
	\end{equation*}
	The corresponding local Yang--Baxter equation 
	\begin{equation} \label{eq:lYB-A}
	L_{12}^A (x_1) L_{13}^A  (x_2) L^A_{23}(x_3) = 
	L_{23}^A(\tilde{x}_3) L_{13}^A(\tilde{x}_2) L_{12}^A(\tilde{x}_1),
	\end{equation}
	reads explicitly as follows
	\begin{equation*}
	\left( \begin{array}{ccc} x_1 x_2 & x_1(1 - x_3) & 1 \\ (1-x_1) x_2 & (1 - x_1)(1-x_3) & 0 \\ 
	1 - x_2 & 0 & 0\end{array} \right) = 
	\left( \begin{array}{ccc} \tilde{x}_2 \tilde{x}_1 & \tilde{x}_2 & 1 \\ 
	(1-\tilde{x}_2) \tilde{x}_1 + \tilde{x}_3(1-\tilde{x}_1) & 1 -\tilde{x}_2  & 0 \\ 
	(1-\tilde{x}_3)  (1-\tilde{x}_1) & 0 & 0\end{array} \right),
	\end{equation*} 
	and is equivalent to the transformation formulas~\eqref{eq:F-map}.
\end{Prop}

The final result of this Section reads as follows.
\begin{Prop} \label{prop:F-Z}	
	The non-commutative affine Hirota map given by \eqref{eq:F-map} satisfies the functional Zamolodchikov equation.
\end{Prop}
This result can be proven by direct verification again, but the calculations become more involved then in the case of the Hirota map. Notice however that the geometric content of both maps is the same. Similarly, the technique via the local Yang--Baxter equation transitions, which was applied in the previous Section, works also here without modifications. 
In next Sections we provide more conceptual geometric approach to Zamolodchikov property of the affine Hirota map via its decomposition into simpler maps.

\section{Non-commutative pentagonal maps and the ten-term relation}
\label{sec:pentagonal-decomposition}

\subsection{The normalization map}
Following~\cite{DoliwaSergeev-pentagon} consider four collinear points 
$AB$, $AC$, $AD$ and $AE$ on a line $A$. 
Their non-homogeneous coordinates $\bPhi_{AK}$, $K=B,C,D,E$ satisfy two pairs of linear relations
\begin{equation} \label{eq:lin-N}
\begin{array}{ll}
\bPhi_{AB}- \bPhi_{AC} = (\bPhi_{AD} - \bPhi_{AC})x_1 \\
\bPhi_{AE} -\bPhi_{AB} = (\bPhi_{AD} - \bPhi_{AB})x_2 
\end{array}
\qquad \text{and} \qquad
\begin{array}{ll}
\bPhi_{AB} - \bPhi_{AC} = (\bPhi_{AE} - \bPhi_{AC}) \hat{x}_1  \\
\bPhi_{AE} - \bPhi_{AB} = (\bPhi_{AD} - \bPhi_{AB} ) \hat{x}_2 .
\end{array}
\end{equation}
The normalization map $N\colon (x_1,x_2) \dashrightarrow  (\hat{x}_1, \hat{x}_2) $ is defined as a consequence of that change of basis
\begin{equation} \label{eq:N}
\begin{array}{ll}
\hat{x}_1 = & \left[ x_2 + x_1(1-x_2)\right]^{-1}x_1, \\
 \hat{x}_2 = & x_2 + x_1(1-x_2),
\end{array}
\end{equation}
i.e. equations \eqref{eq:lin-N} can be considered as a linear problem for the map.
\begin{Rem}
	In \cite{DoliwaSergeev-pentagon} the normalization map was described in homogeneous coordinates which resulted in different form of the transformation, but the geometric meaning is the same (see Section~\ref{sec:NP-hom} for more information).
\end{Rem}
\begin{Prop} \label{prop:N-p}
The normalization map satisfies the functional pentagonal equation
\begin{equation} \label{eq:N-pent}
N_{12} \circ N_{13} \circ  N_{23} = N_{23} \circ  N_{12} .
\end{equation}
\end{Prop}
\begin{proof}
	By direct verification. We provide again only the final expressions
	\begin{equation*}
	(x_1^\prime , x_2^\prime ,x_3^\prime  ) = 
	\begin{cases}
	(N_{12} \circ N_{13} \circ  N_{23})(x_1,x_2,x_3) ,\\ 
	\; \; \; (N_{23} \circ  N_{12})	(x_1,x_2,x_3),
	\end{cases} 
	\end{equation*}
	where
	\begin{align*}
	x_1^\prime & = 
		\left[ x_2 + x_1(1-x_2)\right]^{-1}x_1 , \\
	x_2^\prime & = \left[x_3 + x_2 (1 - x_3) + x_1(1-x_2)(1-x_3) \right]^{-1}
	\left[ x_2 + x_1(1-x_2)\right] ,
	\\
	x_3^\prime & = x_3 + x_2 (1 - x_3) + x_1(1-x_2)(1-x_3).
	\end{align*}
\end{proof}
\begin{Rem}
	The pentagonal condition \eqref{eq:N-pent} looks like a restricted version of the more famous Yang--Baxter condition~\eqref{eq:f-YB}.  However, in modern theory of quantum groups \cite{BaajSkandalis,Woronowicz-P,Timmermann} the quantum pentagon equation seems to play more profound role. See also~\cite{Zakrzewski,STS} for discussion of pentagonal property of Poisson maps.
\end{Rem}

\begin{Cor}
The normalization map $N$ is birational with its inverse given by
\begin{equation} \label{eq:N-inv}
\begin{array}{ll}
x_1 & = \hat{x}_2 \hat{x}_1 , \\ 
x_2 & = (1-\hat{x}_2 \hat{x}_1)^{-1} \hat{x}_2 (1-\hat{x}_1).
\end{array}
\end{equation}
\end{Cor}

Let us present a combinatorial interpretation, slightly different from that given in \cite{DoliwaSergeev-pentagon}, of the pentagonal property of the normalization map. A single linear relation
\begin{equation*}
\bPhi_{AB} - \bPhi_{AD} = (\bPhi_{AC} - \bPhi_{AD} )x  
\end{equation*}
is represented graphically as in Figure~\ref{fig:single}a.
\begin{figure}
\begin{center}
\includegraphics[width=14cm]{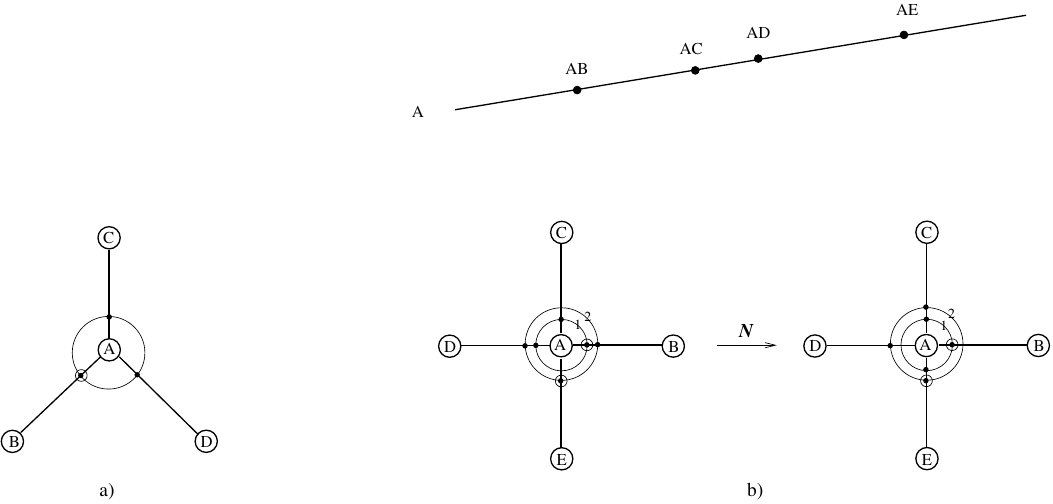}
\end{center}
\caption{Visualization of a single normalized linear relation (a);  graphical representation of the normalization map (b)}
\label{fig:single}
\end{figure}
The collinear points are labeled by two letters and correspond to edges of the graph linked by a circle. The first point $AB$ of the linear equation is marked, and the clockwise orientation $(AB,AC,AD)$ along the circle determines the form of the equation.
The transition from two linear equations \eqref{eq:lin-N} is visualized in Figure~\ref{fig:single}b. 

\begin{figure}[h!]
	\begin{center}
		\includegraphics[width=14cm]{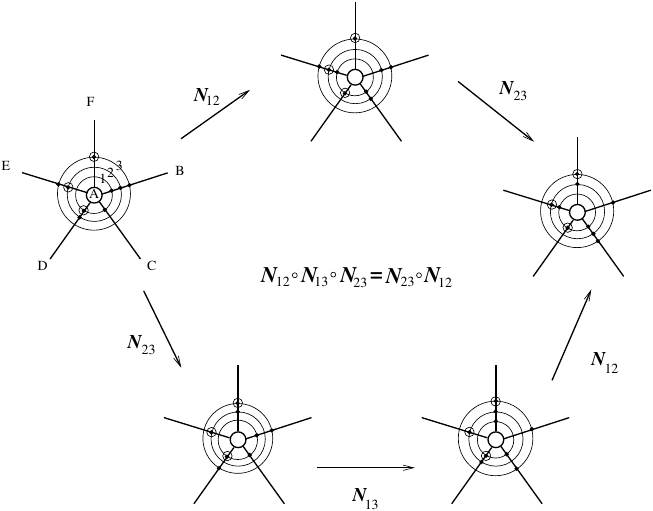}
	\end{center}
	\caption{Pentagonal property of the normalization map}
	\label{fig:pentagon-N}
\end{figure}
The essence of the pentagonal property of the normalization map follows from the observation that given five collinear points and three initial linear relations 
\begin{equation*} 
\begin{array}{ll}
\bPhi_{AD}- \bPhi_{AC} = (\bPhi_{AB} - \bPhi_{AC})x_1 , \\
\bPhi_{AE} -\bPhi_{AD} = (\bPhi_{AB} - \bPhi_{AD})x_2 , \\
\bPhi_{AF} -\bPhi_{AE} = (\bPhi_{AB} - \bPhi_{AE})x_3 ,
\end{array}
\end{equation*}
there are two different but consistent ways to obtain three other relations
\begin{equation*} 
\begin{array}{ll}
\bPhi_{AD} - \bPhi_{AC} = (\bPhi_{AE} - \bPhi_{AC}) x_1^\prime  , \\
\bPhi_{AE} - \bPhi_{AC} = (\bPhi_{AF} - \bPhi_{AC} ) x_2^\prime , \\
\bPhi_{AF} - \bPhi_{AC} = (\bPhi_{AB} - \bPhi_{AC} ) x_3^\prime ,
\end{array}
\end{equation*} 
as visualized on Figure~\ref{fig:pentagon-N}.
Notice that because of the linear equations, which determine the normalization map, are fixed by the combinatorics of the above diagrams, the combinatorial consistence of Figure~\ref{fig:pentagon-N} provides another proof of Proposition~\ref{prop:N-p}.

\subsection{The Veblen map}
In~\cite{DoliwaSergeev-pentagon} another geometric map was presented as well. The map is related to the Veblen configuration $(6_2,4_3)$ consisting of six points and four lines with two lines through each point and three points on each line, see Figure~\ref{fig:Veblen}.  \begin{Rem}
The Veblen configuration of projective geometry~\cite{BeukenhoutCameron-H} was called the Menelaus configuration in \cite{KoSchief-Men}, where its connection with the discrete Schwarzian KP equation (another algebraic description of Desargues lattices) has been found. The name of the configuration reflects its metric geometry meaning described in the Menelaus theorem~\cite{Coxeter-IG}.
\begin{figure}[h!]
	\begin{center}
		\includegraphics[width=11cm]{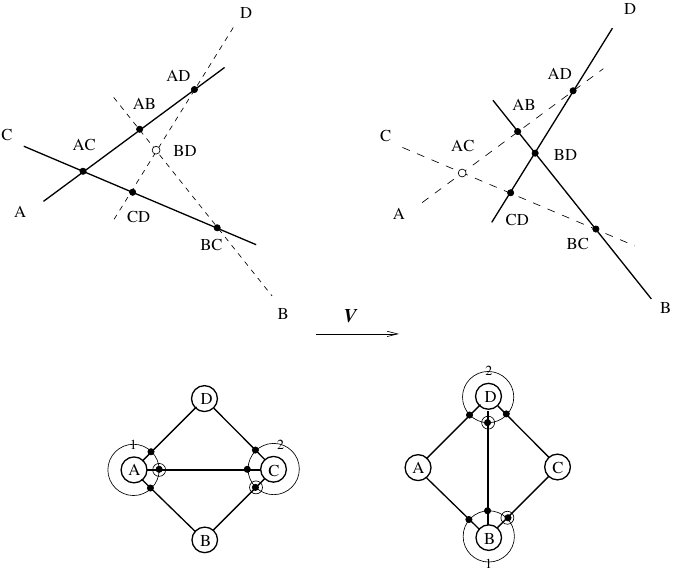}
	\end{center}
	\caption{The Veblen flip and graphical representation of the Veblen map}
	\label{fig:Veblen}
\end{figure}
\end{Rem}	
 Starting from two lines $A$ and $C$ intersecting in one point $AC$ and two additional (labeled) points on each of the line --- $AB$ and $AD$ on the line $A$, and $BC$ and $CD$ on $C$ --- one can construct uniquely two other lines and the sixth point $BD$ of the configuration. In terms of the non-homogeneous coordinates of the points we have
\begin{equation} \label{eq:lin-V}
\begin{array}{cc}
\bPhi_{AC} - \bPhi_{AD} = (\bPhi_{AB} - \bPhi_{AD}) x_1  \\
\bPhi_{BC} - \bPhi_{CD} = (\bPhi_{AC} -\bPhi_{CD}) x_2 
\end{array}
\qquad \text{and} \qquad
\begin{array}{cc}
\bPhi_{BC} -\bPhi_{BD} = (\bPhi_{AB} - \bPhi_{BD})\bar{x}_1  \\
\bPhi_{BD}- \bPhi_{CD} = (\bPhi_{AD}- \bPhi_{CD} )\bar{x}_2 .
\end{array}
\end{equation}
Solving these equation for $\bar{x}_1, \bar{x}_2$ (and $\bPhi_{BD}$) we obtain the Veblen map $V\colon (x,y) \dashrightarrow (\bar{x}, \bar{y}) $, where
\begin{equation} \label{eq:V}
\begin{array}{ll}
\bar{x}_1 & =x_1 x_2, \\ \bar{x}_2 &= (1-x_1) x_2 (1-x_1 x_2)^{-1} .
\end{array}
\end{equation}
The graphical representation of the Veblen map, given also in Figure~\ref{fig:Veblen}, is similar to that used in~\cite{DoliwaSergeev-pentagon}. We provide additional information on the precise form of the linear relations, as given in the corresponding graphical representation of the normalization map, in order to be able to discuss in Section~\ref{sec:tt} simultaneously both maps, whose geometric meanings are different.

\begin{Prop} \label{prop:V-p}
	The Veblen map $V\colon (x_1,x_2) \dashrightarrow (\bar{x}_1, \bar{x}_2) $
	satisfies the reversed pentagonal condition
	\begin{equation} \label{eq:V-pent-r}
	V_{23} \circ V_{13} \circ  V_{12} = V_{12} \circ  V_{23} .
	\end{equation}
\end{Prop}
\begin{proof}
	By direct verification. We provide again only the final expressions
	\begin{equation*}
	(x_1^\prime , x_2^\prime ,x_3^\prime  ) = 
	\begin{cases}
	(V_{23} \circ V_{13} \circ  V_{12})(x_1,x_2,x_3) , \\ 
	\; \; \; (V_{12} \circ  V_{23})	(x_1,x_2,x_3) ,
	\end{cases}
	\end{equation*}
	where
	\begin{align*}
	x_1^\prime & = 
	x_1 x_2 x_3 , \\
	x_2^\prime & = (1 - x_1) x_2 x_3 (1-x_1 x_2 x_3)^{-1} ,
	\\
	x_3^\prime & =  (1 - x_2) x_3(1-x_2 x_3)^{-1}.
	\end{align*}
\end{proof}
\begin{Cor}
	The inverse of $V$ is given by
	\begin{equation} \label{eq:V-inv}
	\begin{array}{ll}
	x_1 &= \bar{x}_1 [ \bar{x}_2 + (1 - \bar{x}_2)\bar{x}_1 ]^{-1} , \\ 
	x_2 & =  \bar{x}_2 + (1 - \bar{x}_2)\bar{x}_1,
	\end{array}
	\end{equation}
	and satisfies the pentagon condition.
\end{Cor}
\begin{Rem}
	By comparison of formulas \eqref{eq:N} with \eqref{eq:V-inv}, and by comparison of  
	formulas \eqref{eq:V} with \eqref{eq:N-inv} we see that 
	\begin{equation}
	V^{-1} = N^\mathrm{op} \qquad \text{and} \qquad N^{-1} = V^\mathrm{op}.
	\end{equation}
\end{Rem}

The geometric content of the pentagon property of the normalization map is rather trivial because it reflects various ways of organizing five (collinear) points into ordered triplets. 
As it was discussed in \cite{DoliwaSergeev-pentagon} the geometric meaning of the pentagonal property of the Veblen map is provided by the Desargues theorem, which defines configuration $(10_3,10_3)$ consisting of ten points and ten lines with three lines through each point and three points on each line, see Figure~\ref{fig:Desargues-pent}. 
\begin{figure}[h!]
\begin{center}
\includegraphics[width=10cm]{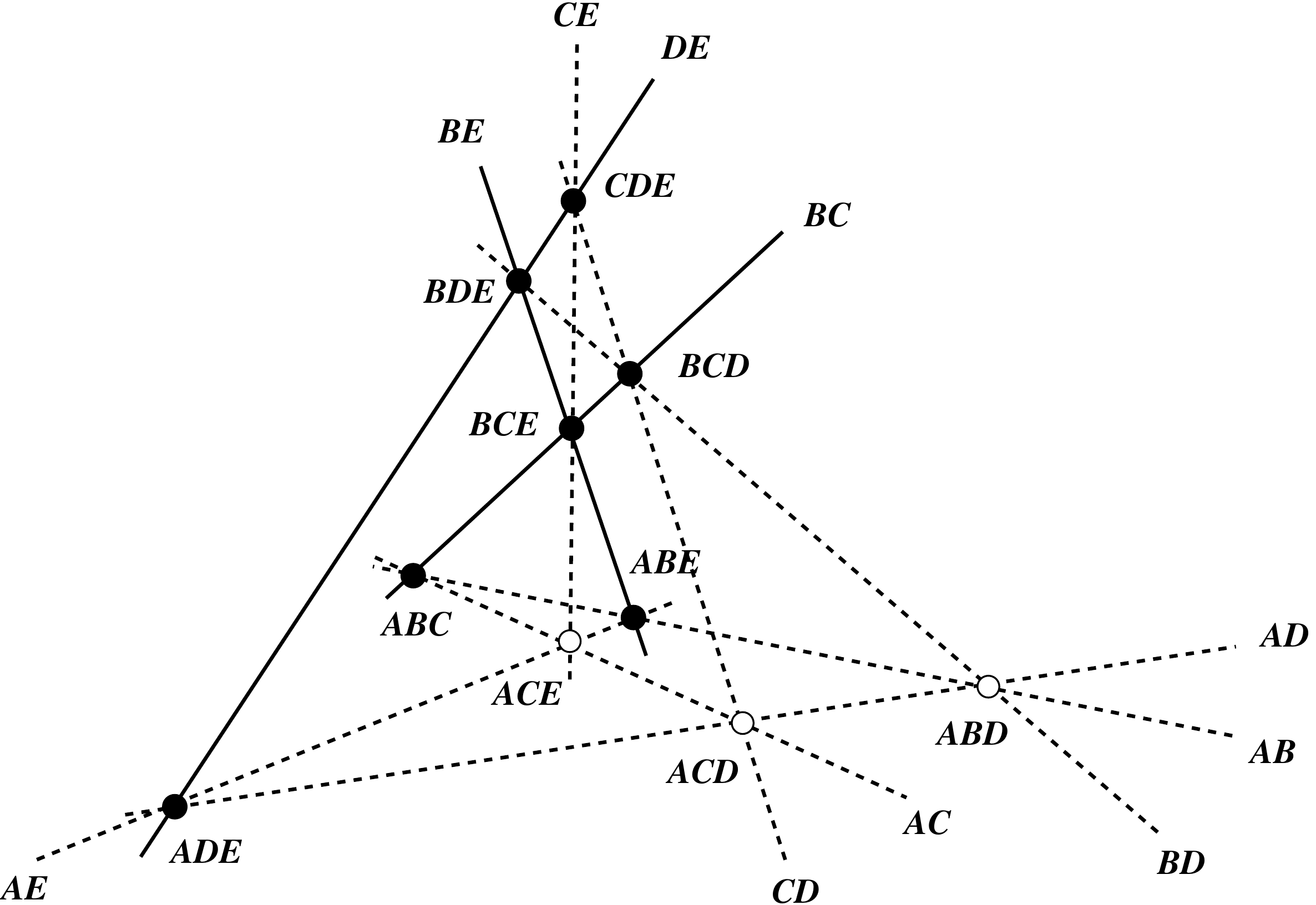}
\end{center}
\caption{Desargues configuration with the initial subconfiguration}
\label{fig:Desargues-pent}
\end{figure}
Its lines are labeled by two-element subsets out of the five-letter set $\{ A,B,C,D,E \}$, while points of the configuration are labeled by three-element subsets (a point belongs to a line when the line-subset is contained in the point-subset). 
We start from initial subconfiguration of three lines (the solid lines on Figure~\ref{fig:Desargues-pent}) and seven (black) points, which gives three linear relations. Using Veblen flips we construct a similar subconfiguration in two different ways. 

\subsection{The ten-term relation} \label{sec:tt}
It is easy to observe that the geometric configuration behind the Desargues map (Figure~\ref{fig:Desargues-map}) is composed with the geometric configurations describing the normalization map and the Veblen map (Figures~\ref{fig:single} and~\ref{fig:Veblen}). 
Below we present the corresponding decomposition formula, which can be verified by direct calculation.
\begin{Prop}
	The  map $F\colon (x_1,x_2,x_3) \dashrightarrow (\tilde{x}_1, \tilde{x}_2 , \tilde{x}_3)$ given by \eqref{eq:F-map}, which  describes algebraically the geometric Desargues map in non-homogeneous coordinates, can be factorized in terms of the normalization and Veblen maps as follows
	\begin{equation}
	\label{eq:fact-Z}
	F = P_{23} \circ V_{12} \circ N_{13},
	\end{equation}
	where $P_{23}$ is the transposition of the second and third arguments.
\end{Prop}
\begin{proof}
	\begin{gather*}
	\left( \begin{array}{c} x_1 \\ x_2 \\ x_3 \end{array} \right) \stackrel{N_{13}}{\longrightarrow}
	\left( \begin{array}{l} \hat{x}_1 = \left[ x_3 + x_1(1-x_3)\right]^{-1}x_1 \\
	\hat{x}_2 = x_2\\ 
	\hat{x}_3  =  x_3 + x_1(1-x_3) \end{array} \right) \stackrel{V_{12}}{\longrightarrow}
\left( \begin{array}{l} 
\hat{x}_1 \hat{x}_2  \\  
(1-\hat{x}_1) \hat{x}_2 (1-\hat{x}_1 \hat{x}_2)^{-1} 
\\  \hat{x}_3  
\end{array} \right) = \\
\left( \begin{array}{r} 
\left[ x_3 + x_1(1-x_3)\right]^{-1}x_1 x_2  = \tilde{x}_1\\  
1 + (x_2-1) \left[ (1-x_1) x_3 + x_1(1-x_2)\right]^{-1} \left( x_3 + x_1 (1-x_3) \right)
= \tilde{x}_3
\\   x_3 + x_1(1-x_3) = \tilde{x}_2
\end{array} \right)
\stackrel{P_{23}}{\longrightarrow}	
\left( \begin{array}{c} \tilde{x}_1 \\ \tilde{x}_2 \\ \tilde{x}_3 \end{array} \right)
	\end{gather*}
\end{proof}
Construction of solutions of the Zamolodchikov equation using solutions of the pentagonal equation was given in~\cite{Maillet}.
Following~\cite{Kashaev-Sergeev} we state (and recall the proof for completeness) the so called ten-term relation which is an important step in proving that the map $F$ given by the decomposition \eqref{eq:fact-Z} satisfies the functional Zamolodchikov equation.
\begin{Lem} \label{lem:tt}
	Assume that the invertible map $N$ satisfies the pentagonal equation, and the invertible map $V$ satisfies the reversed pentagonal equation. Then the map $F$, given by the decomposition \eqref{eq:fact-Z} satisfies the functional Zamolodchikov equation, if and only if the maps $N$ and $V$ are paired by the following ten-term relation
\begin{equation} \label{eq:ten-term}
V_{13}\circ N_{12}\circ V_{14}\circ N_{34}\circ V_{24} = N_{34}\circ V_{24}\circ N_{14}\circ V_{13}\circ N_{12} . 
\end{equation}	
\end{Lem}
\begin{proof}
	Assuming the factorization \eqref{eq:fact-Z} in the Zamolodchikov equation
			\begin{equation*}
	F_{123} \circ F_{145} \circ F_{246}\circ  F_{356} = F_{356} \circ F_{246} \circ F_{145} \circ F_{123},
	\end{equation*}
	and shifting all permutations to the right we obtain the relation
	\begin{equation*}
	\left( V_{16}^{-1} \circ V_{36}^{-1} \circ V_{13} \circ V_{36} \right) 
	\circ N_{12} \circ V_{15} \circ N_{35} \circ V_{25} = N_{35} \circ V_{25} \circ N_{15} \circ V_{13} \circ N_{12} \circ
	\left( N_{24} \circ N_{12} \circ N_{24}^{-1} \circ N_{14}^{-1} \right) ,
	\end{equation*}
	which gives the ten-term relation \eqref{eq:ten-term} in variables $(x_1,x_2,x_3,x_5)$ due to the fact that $N$ satisfies --- here in variables $(x_1,x_2,x_4)$ ---  the pentagonal equation~\eqref{eq:N-pent}, and  $V$ satisfies --- here in variables $(x_1,x_3,x_6)$--- the reversed pentagonal equation~\eqref{eq:V-pent-r}.  
\end{proof}
\begin{Prop} \label{prop:NV-tt}
The normalization map \eqref{eq:N} and the Veblen map \eqref{eq:V} are paired by the ten-term relation \eqref{eq:ten-term}.
\end{Prop}
\begin{proof}
We will demonstrate the result in two ways. The first one is by direct (but slightly more tedious then before) calculation. As usual we provide the final expressions only
\begin{equation*}
(x_1^\prime , x_2^\prime ,x_3^\prime , x_4^\prime ) = 
\begin{cases}
(V_{13}\circ N_{12}\circ V_{14}\circ N_{34}\circ V_{24})(x_1,x_2,x_3,x_4) ,\\ 
(N_{34}\circ V_{24}\circ N_{14}\circ V_{13}\circ N_{12})	(x_1,x_2,x_3,x_4),
\end{cases}
\end{equation*}
where
\begin{align*}
x_1^\prime & = 
\left( (x_2 + x_1(1-x_2))x_4 + x_1 x_3 (1-x_4)\right)^{-1}x_1 x_3 ,\\
x_2^\prime & =  (x_2 + x_1(1-x_2))x_4 + x_1 x_3 (1-x_4) ,
\\
x_3^\prime & = (1-x_2 x_4) ((1-x_2)x_4 + x_3  (1 - x_4))^{-1} 
\left( 1 + (x_3 - 1) \left(  (1-x_1) x_2 + x_1 (1-x_3)   \right)^{-1}x_1
\right)x_3 ,
 \\
x_4^\prime & =  ((1-x_2)x_4 + x_3  (1 - x_4))(1-x_2 x_4)^{-1} .
\end{align*}

The second way exploits the geometric meaning of the ten-term pairing between the 
normalization and the Veblen maps. Consider the (star) configuration  consisting of ten points and five lines with two lines through each point and four points on each line, as visualized on Figure~\ref{fig:star}. 
\begin{figure}
	\begin{center}
		\includegraphics[width=7cm]{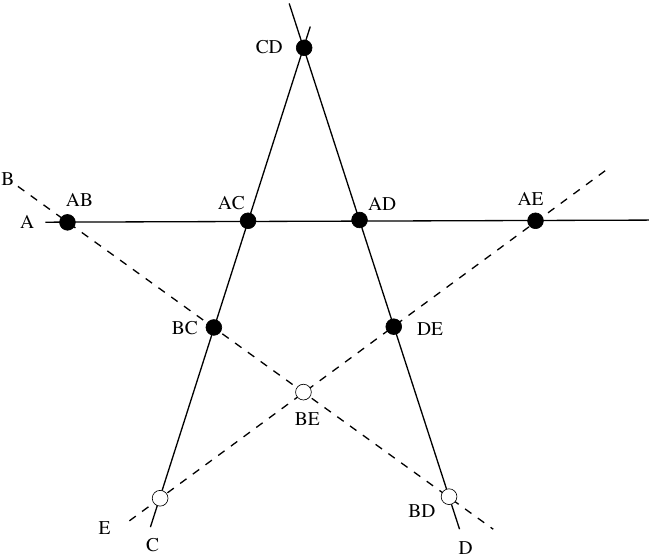}
	\end{center}
	\caption{The star  configuration $(10_2,5_4)$ as geometric origin of the ten-term relation between the normalization and Veblen maps}
	\label{fig:star}
\end{figure}
Let us start with seven points $AB, AC, AD, AE, BC, CD, DE$ of the configuration (black circles on Figure~\ref{fig:star}) and four corresponding linear relations 
\begin{equation*} 
\begin{array}{ll}
\bPhi_{AD}- \bPhi_{AE} = (\bPhi_{AB} - \bPhi_{AE})x_1 ,\\
\bPhi_{AC} -\bPhi_{AD} = (\bPhi_{AB} - \bPhi_{AD})x_2 , \\
\bPhi_{CD} -\bPhi_{DE} = (\bPhi_{AD} - \bPhi_{DE})x_3 , \\
\bPhi_{BC} -\bPhi_{CD} = (\bPhi_{AC} - \bPhi_{CD})x_4 ,
\end{array}
\end{equation*}
as described combinatorially on the leftmost pentagon of Figure~\ref{fig:ten-term}.
There are two distinct ways to obtain (using the normalization and Veblen flips) another such subconfiguration formed by seven points $AB, AE, BC, BE, CD, CE, DE$ whose non-homogeneous coordinates satisfy four other linear relations
\begin{equation*} 
\begin{array}{ll}
\bPhi_{CD} - \bPhi_{CE} = (\bPhi_{BC} - \bPhi_{CE}) x_1^\prime , \\
\bPhi_{BC} - \bPhi_{BE} = (\bPhi_{AB} - \bPhi_{BE} ) x_2^\prime , \\
\bPhi_{CE} - \bPhi_{DE} = (\bPhi_{BE} - \bPhi_{DE} ) x_3^\prime , \\
\bPhi_{BE} - \bPhi_{DE} = (\bPhi_{AE} - \bPhi_{DE} ) x_4^\prime ,
\end{array}
\end{equation*} 
as described combinatorially on the rightmost pentagon of Figure~\ref{fig:ten-term}.
The consecutive transformations, whose combinatorial description is visualized on Figures~\ref{fig:single} and~\ref{fig:Veblen}, are presented on Figure~\ref{fig:ten-term} and the combinatorial consistency of the two different transitions  provides another proof of ten-term relation~\eqref{eq:ten-term}.
\end{proof}
\begin{Rem}
	The geometric configurations behind the pentagonal property of the normalization map, the ten-term relation pairing the normalization and Veblen map, and the pentagonal property of the Veblen map, exhibit $\mathcal{S}_5$ symmetry of the permutation group of five elements. In the first case of five collinear points this observation is trivial. In the second case of the star configuration $(10_2,5_4)$, and the third case of the Desargues configuration $(10_3,10_3)$, as we have already mentioned, the points and lines are labeled by subsets of the five element set. 
	
	It turns out~\cite{Dol-AN} that it is more convenient to define $N$-dimensional Desargues maps on the $A_N$ root lattice, instead of the standard $N$-dimensional integer lattice $\ZZ^N$. The corresponding combinatorial and group structures of the root lattice allow to see the symmetry structure of the Hirota--Miwa system from the very beginning. In particular, the first ``configuration'' of five collinear points  is the image of vertices of the Delunay tile $P(1,4)$ of the lattice under the Desargues map, and the star configuration and the Desargues configuration are the images of vertices of the tiles $P(2,4)$ and $P(3,4)$, respectively (see~\cite{Dol-AN} for more details). 
\end{Rem}
\begin{Cor}
	We concluded therefore the geometro-combinatorial proof of Proposition~\ref{prop:F-Z} using decomposition of the map $F$ into the normalization and Veblen maps.	
\end{Cor}
\begin{figure}
	\begin{center}
		\includegraphics[width=16cm]{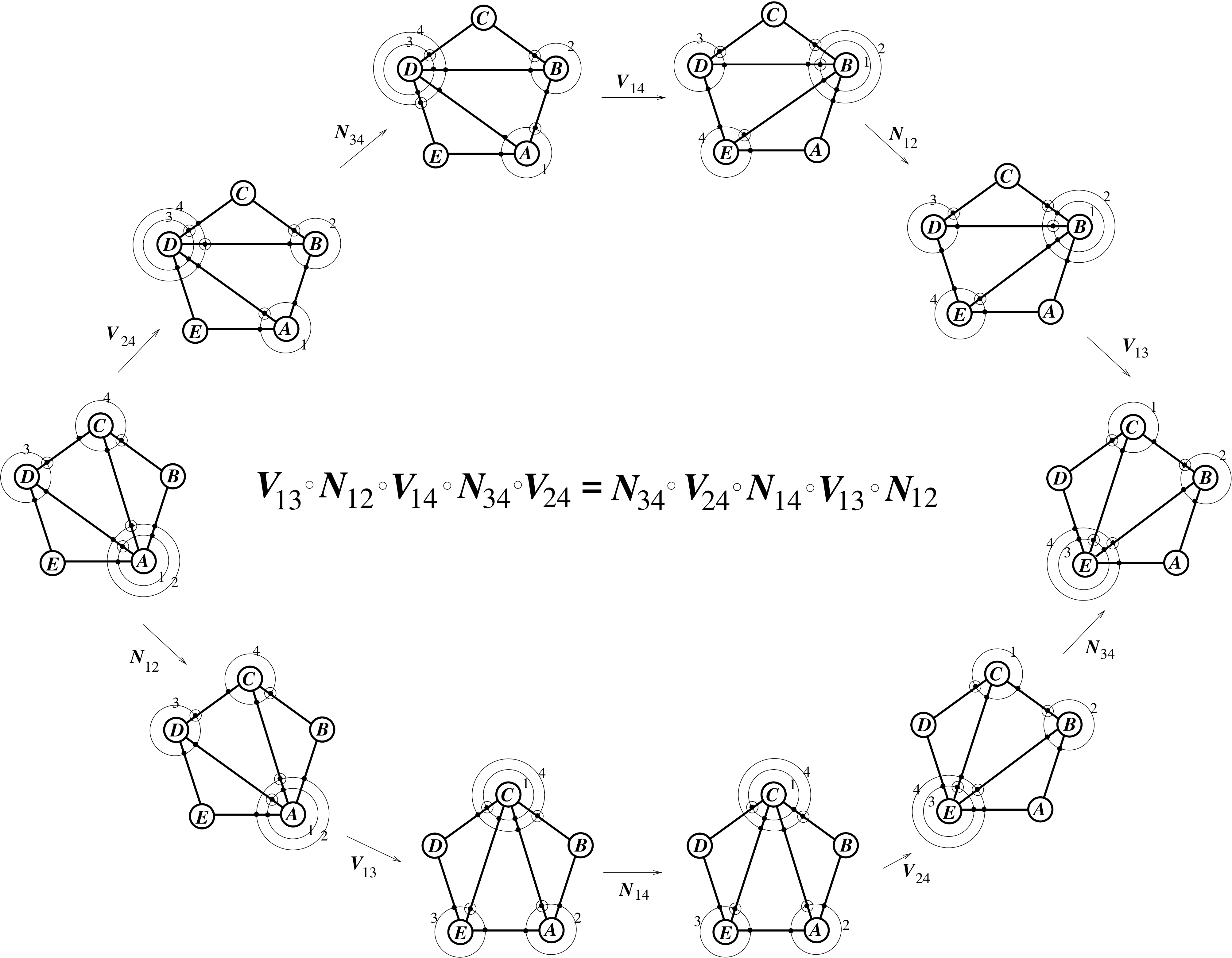}
	\end{center}
	\caption{The graphical representation of the ten-term relation}
	\label{fig:ten-term}
\end{figure}
\begin{Rem}
	The Hirota map~\eqref{eq:Hirota-map} also allows for the decomposition $H = P_{23} \circ T_{12} \circ S_{23}$ where
	\begin{align*}
	S(y_1,y_2) & = ( (y_1 + y_2)^{-1}y_1, y_1 + y_2), \\
	T(y_1,y_2) & = ( y_1 y_2, (1- y_1) y_2) .
	\end{align*}
The maps $S$ and $T$ can be given an interpretation as a normalization map and Veblen map, respectively. However the linear problems for the constituent maps mix the affine and Hirota gauges, what makes this decomposition less useful then the previous one.
\end{Rem}

\section{Quantum and Poisson reductions of the normalization and Veblen maps}
\label{sec:NP-hom}
We first recall description of the normalization and Veblen maps in arbitrary homogeneous coordinates, as it was done in~\cite{DoliwaSergeev-pentagon}. This approach allows to give the ultra-local quantum reductions of the maps. Their classical limit gives the Poisson reductions, also studied in ~\cite{DoliwaSergeev-pentagon}. Then we apply these results to present the corresponding maps which satisfy the Zamolodchikov equation.  
Because in our presentation we take sometimes different ordering of points of the geometric configurations behind the normalization and Veblen maps our formulas may differ slightly from those given in~\cite{DoliwaSergeev-pentagon}. We remark that an example of pentagonal map in non-commuting variables was given for the first time in~\cite{Kashaev-Reshetikhin-P}.

\subsection{The normalization map in homogeneous coordinates}
Consider the full homogeneous version of the linear problem~\eqref{eq:lin-N} 
\begin{equation} \label{eq:lin-N-h}
\begin{array}{ll}
\bpsi_{AB} = \bpsi_{AD}x_1 + \bpsi_{AC} y_1 \\
\bpsi_{AE} = \bpsi_{AD}x_2 + \bpsi_{AB} y_2 
\end{array}
\qquad \text{and} \qquad
\begin{array}{ll}
\bpsi_{AB}  = \bpsi_{AE}\hat{x}_1  + \bpsi_{AC} \hat{y}_1 ,  \\
\bpsi_{AE} = \bpsi_{AD} \hat{x}_2 + \bpsi_{AC} \hat{y}_2 .
\end{array}
\end{equation}
which gives the corresponding normalization map $\mathcal{N}\colon [(x_1,y_1),(x_2,y_2)] \dashrightarrow  [(\hat{x}_1,\hat{y}_1), (\hat{x}_2,\hat{y}_2)] $ 
\begin{equation} \label{eq:N-h}
\begin{array}{ll}
\hat{x}_1 = & \left( x_2 + x_1y_2\right)^{-1}x_1, \\
\hat{x}_2 = & x_2 + x_1 y_2,
\end{array} \qquad
\begin{array}{ll}
\hat{y}_1 = & y_1 x_1^{-1} x_2 \left( x_2 + x_1 y_2 \right)^{-1}x_1 \\
\hat{y}_2 = & y_1 y_2
\end{array}
\end{equation}
Notice that the combinatorial description of Figure~\ref{fig:single} of the linear problems fixes uniquely the form of the equations. Any linear equation corresponds to a circle (its label gives indices of the fields entering the equation) surrounding a vertex of the graph. There are three intersection points of the circle with edges incident with the vertex: the encircled one gives the label of the wave function on the left hand side of the equation while the right and the left points give labels of the wave functions multiplied by the fields $x$ and $y$, respectively.
\begin{Prop}
The homogeneous normalization map \eqref{eq:N-h} is birational with the inverse given by
\begin{equation} \label{eq:N-h-inv}
\begin{array}{ll}
x_1 = & \hat{x}_2 \hat{x}_1, \\
x_2 = & \hat{x}_2 \hat{y}_2^{-1} \hat{y}_1 (\hat{y}_2 \hat{x}_1 + \hat{y}_1)^{-1} \hat{y}_2 ,
\end{array} \qquad
\begin{array}{ll}
y_1 = &  \hat{y}_2 \hat{x}_1 + \hat{y}_1 ,\\
y_2 = & (\hat{y}_2 \hat{x}_1 + \hat{y}_1)^{-1} \hat{y}_2 .
\end{array}
\end{equation}
\end{Prop}
\begin{Rem}
	Notice that in the homogeneous case to obtain the inverse of the normalization map we exchange both the letters $x\leftrightarrow y$ and indices $1\leftrightarrow 2$.
\end{Rem}
\begin{Prop}
	The homogeneous normalization map allows for the affine/non-homogeneous reduction $x_i+y_i=1$, $i=1,2$, giving the previous formulas~\eqref{eq:N}.
\end{Prop}
\begin{proof}
	We have to check that if the initial fields satisfy the affine reduction condition then also the final fields given by \eqref{eq:N-h} satisfy the condition as well. This is obvious from geometric point of view, but the interested Reader can verify it directly.
\end{proof}
\begin{Prop} \label{prop:N-p-h}
	The homogeneous normalization map $\mathcal{N}$ satisfies the functional pentagonal equation
	\begin{equation} \label{eq:N-pent-h}
	\mathcal{N}_{12} \circ \mathcal{N}_{13} \circ  \mathcal{N}_{23} = \mathcal{N}_{23} \circ  \mathcal{N}_{12} .
	\end{equation}
\end{Prop}
\begin{proof}
	To make our presentation complete we provide the final expressions
	\begin{equation*}
	[(x_1^\prime , y_1^\prime) , (x_2^\prime, y_2^\prime) , (x_3^\prime , y_3^\prime )] = 
	\begin{cases}
	(\mathcal{N}_{12} \circ \mathcal{N}_{13} \circ  \mathcal{N}_{23})[(x_1,y_1),(x_2,y_2), (x_3,y_3)] ,\\ 
	\; \; \; (\mathcal{N}_{23} \circ  \mathcal{N}_{12})	[(x_1,y_1),(x_2,y_2), (x_3,y_3)] ,
	\end{cases}
	\end{equation*}
	where
	\begin{align*}
	x_1^\prime & = (x_2 + x_1 y_2)^{-1} x_1 , \\
	y_1^\prime & =  y_1 x_1^{-1} x_2 (x_2 + x_1 y_2)^{-1} x_1 , \\
	x_2^\prime & = \left[x_3 + (x_2 + x_1 y_2) y_3\right]^{-1} (x_2 + x_1 y_2), \\
	y_2^\prime & = y_1 y_2 \left[ 1 + y_3 x_3^{-1}(x_2 + x_1 y_2)\right]^{-1}, \\
	x_3^\prime & = x_3 + (x_2 + x_1 y_2) y_3,\\
	y_3^\prime & = y_1 y_2 y_3 .
	\end{align*}
\end{proof}
\begin{Rem}
The combinatorial proof visualized on Figure~\ref{fig:pentagon-N} keeps its value also in the homogeneous gauge.
\end{Rem}

In \cite{DoliwaSergeev-pentagon} the ultra-local reduction of the (direct analogue) of the functional pentagon map \eqref{eq:N-h} was given. The ultra-locality means that the fields with different indices commute, i.e. in our case we assume that
\begin{equation} \label{eq:u-l}
x_1 x_2 = x_2 x_1, \qquad y_1 y_2 = y_2 y_1, \qquad
x_1 y_2 = y_2 x_1, \qquad y_1 x_2 = x_2 y_1. 
\end{equation}
By direct calculation one shows~\cite{DoliwaSergeev-pentagon} the following result.
\begin{Prop}
	The homogeneous normalization map \eqref{eq:N-h} allows for the ultra-local reduction
	provided that there exists central non-vanishing constant $q$ such that the fields $[(x_1,y_1),(x_2,y_2)]$ satisfy, apart from the ultra-locality requirements \eqref{eq:u-l}, also the Weyl commutation relations
	\begin{equation}
	\label{eq:Weyl-c-r}
	x_1 y_1 = q y_1 x_1, \qquad x_2 y_2 = q y_2 x_2.
	\end{equation}
	In that case the resulting fields  $[(\hat{x}_1,\hat{y}_1), (\hat{x}_2,\hat{y}_2)]$ given by~\eqref{eq:N-h} satisfy not only the reduction conditions \eqref{eq:u-l}, but also the Weyl commutation relations
	\eqref{eq:Weyl-c-r} with the same constant $q$.
\end{Prop}
\begin{Rem}
	In the ultra-local/quantum reduction the normalization map \eqref{eq:N-h} simplifies slightly and can be rewritten in the form
	\begin{equation} \label{eq:N-q}
	\begin{array}{ll}
	\hat{x}_1 = & x_1\left( x_2 + x_1y_2\right)^{-1}, \\
	\hat{x}_2 = & x_2 + x_1 y_2,
	\end{array} \qquad
	\begin{array}{ll}
	\hat{y}_1 = & y_1 x_2 \left( x_2 + x_1 y_2 \right)^{-1} \\
	\hat{y}_2 = & y_1 y_2.
	\end{array}
	\end{equation}
\end{Rem}
\begin{Rem}
	The above procedure allows to find a transition to solutions of the quantum pentagon (and then also Zamolodchikov's) equations. The important step, which is outside of the scope of the present work (see however~\cite{Sergeev-book,Sergeev-q3w,Sergeev-finite,DoliwaSergeev-pentagon}), is  to consider an irreducible (preferably) representation of the division algebra generated by relations~\eqref{eq:u-l}-\eqref{eq:Weyl-c-r} and find a realization of the  automorphism~\eqref{eq:N-q} in terms of an inner map.  
\end{Rem}
\begin{Cor}
	By taking quasi-classical $q\to 1$ limit of the above commutation relations we obtain that for commuting variables the functional pentagon map \eqref{eq:N-h} is a Poisson map with the brackets
	\begin{equation} \label{eq:Poisson-Weyl}
	\{ x_1, y_1 \} = x_1 y_1, \quad \{ x_2, y_2 \} = x_2 y_2, \quad \{ x_1, x_2 \} = \{ y_1, y_2 \} = \{ x_1, y_2 \} = \{ y_1, x_2 \} =0.
	\end{equation}
	\end{Cor}
\begin{Rem}
	The form~\eqref{eq:N-q} of the quantum normalization map does not change in the quasi-classical limit, therefore we will use the name ultra-local/quantum/Poisson map. The difference --- the Weyl commutation relations~\eqref{eq:Weyl-c-r} versus the Poisson brackets~\eqref{eq:Poisson-Weyl} --- will be clear from the context.
\end{Rem}

\subsection{The Veblen map in homogeneous coordinates}
Following~\cite{DoliwaSergeev-pentagon} let us present the Veblen map in the generic homogeneous gauge. As in the previous case of the normalization map we discuss then various gauge specifications of the map, and finally we will move on to the ultra-local and Poisson reductions. 

The homogeneous version of the linear problem~\eqref{eq:lin-V} reads as follows
\begin{equation} \label{eq:lin-V-h}
\begin{array}{ll}
\bpsi_{AC} = \bpsi_{AB}x_1 + \bpsi_{AD} y_1 \\
\bpsi_{BC} = \bpsi_{AC}x_2 + \bpsi_{CD} y_2 
\end{array}
\qquad \text{and} \qquad
\begin{array}{ll}
\bpsi_{BC}  = \bpsi_{AB}\bar{x}_1  + \bpsi_{BD} \bar{y}_1 , \\
\bpsi_{BD} = \bpsi_{AD} \bar{x}_2 + \bpsi_{CD} \bar{y}_2 .
\end{array}
\end{equation}
Its combinatorial description, as given on Figure~\ref{fig:Veblen}, remains valid exactly like in the case of homogeneous normalization map.

Notice that even if the point $BD$ is uniquely constructed from other five points of the Veblen configuration its homogeneous coordinates are given up to a right non-zero factor. This implies that the value of $\bar{y}_1$  can be given arbitrary. To underline this freedom, without loosing generality we put $\bar{y}_1=G \neq 0$. The corresponding Veblen map $\mathcal{V}^G\colon [(x_1,y_1),(x_2,y_2)] \dashrightarrow  [(\bar{x}_1,\bar{y}_1), (\bar{x}_2,\bar{y}_2)] $ is given by
\begin{equation} \label{eq:V-h}
\begin{array}{ll}
\bar{x}_1 = & x_1 x_2, \\
\bar{x}_2 = & y_1 x_2 G^{-1},
\end{array} \qquad
\begin{array}{ll}
\bar{y}_1 = &G , \\
\bar{y}_2 = & y_2 G^{-1}.
\end{array}
\end{equation}
\begin{Prop}
	The homogeneous normalization map allows for the affine/non-homogeneous reduction $x_i+y_i=1$, $i=1,2$, giving the previous formulas~\eqref{eq:V}.
\end{Prop}
\begin{proof}
	Assume that the initial fields satisfy the affine reduction condition and fix the gauge function requiring
	\begin{equation}
	G= 1 - \bar{x}_1 = 1- x_1 x_2,
	\end{equation}
	Then the remaining parts of equations \eqref{eq:V-h} reduce to those of~\eqref{eq:V}.
\end{proof}
\begin{Rem}
	By solving the linear problem \eqref{eq:lin-V-h} in opposite direction we can analogously multiply the wave function $\bpsi_{AC}$ by an arbitrary right factor, what results in the freedom $x_2=\bar{G}$ which gives
\begin{equation} \label{eq:V-h-inv}
\begin{array}{ll}
x_1 = & \bar{x}_1 \bar{G}^{-1}, \\
x_2 = & \bar{G},
\end{array} \qquad
\begin{array}{ll}
y_1 = & \bar{x}_2 \bar{y}_1  \bar{G}^{-1},  \\
y_2 = & \bar{y}_2 \bar{y}_1  .
\end{array}
\end{equation}	
\end{Rem}

The gauge parameters enter also in the corresponding pentagonal condition.
\begin{Prop} \label{prop:V-p-h}
	The homogeneous Veblen map \eqref{eq:V-h}
	satisfies the reversed pentagonal condition
	\begin{equation} 
	\mathcal{V}_{23}^C \circ \mathcal{V}_{13}^B \circ  \mathcal{V}_{12}^A = \mathcal{V}_{12}^H \circ  \mathcal{V}_{23}^G ,
	\end{equation}
	provided the gauge parameters of the maps satisfy conditions
	\begin{equation} \label{eq:V-pent-g-ab}
	H = B, \qquad G = CB .
	\end{equation}
\end{Prop}
\begin{proof}
	We provide the final expressions
	\begin{equation*}
	[(x_1^\prime , y_1^\prime) , (x_2^\prime, y_2^\prime) , (x_3^\prime , y_3^\prime )] = 
	\begin{cases}
	\mathcal{V}_{23}^C \circ \mathcal{V}_{13}^B \circ  \mathcal{V}_{12}^A[(x_1,y_1),(x_2,y_2), (x_3,y_3)] , \\ 
	\; \; \; (\mathcal{V}_{12}^H \circ  \mathcal{V}_{23}^G)	[(x_1,y_1),(x_2,y_2), (x_3,y_3)] 
	\end{cases}
	\end{equation*}
	where
	\begin{align*}
	x_1^\prime & = x_1 x_2 x_3 , \\
	y_1^\prime & =  H = B , \\
	x_2^\prime & = y_1 x_2 x_3 H^{-1} = y_1 x_2 x_3 B^{-1} , \\
	y_2^\prime & =GH^{-1} = C, \\
	x_3^\prime & = y_2 x_3 G^{-1} = y_2 x_3 B^{-1}C^{-1},\\
	y_3^\prime & = y_3 G^{-1} = y_3 B^{-1} C^{-1} .
	\end{align*}
\end{proof}
\begin{Ex}
	Let us fix the gauge function $G$ by the condition $\bar{y}_1=y_1$, which gives the corresponding specification of the Veblen map
	\begin{equation} \label{eq:V-h-ex}
	\begin{array}{l}
	\bar{x}_1 =  x_1 x_2, \\
	\bar{x}_2 =  y_1 x_2 y_1^{-1},
	\end{array} \qquad
	\begin{array}{l}
	\bar{y}_1 =  y_1 \\
	\bar{y}_2 =  y_2 y_1^{-1}.
	\end{array}
	\end{equation}	
	Such a map satisfies the pentagonal condition.
\end{Ex}

\begin{Prop}
The homogeneous Veblen map \eqref{eq:V-h} allows for the ultra-local reduction
provided that there exists central non-vanishing constant $q$ such that the initial fields $[(x_1,y_1),(x_2,y_2)]$ satisfy, apart from the ultra-locality requirements \eqref{eq:u-l}, also the Weyl commutation relations \eqref{eq:Weyl-c-r},
and when the gauge parameter has the form
\begin{equation} \label{eq:G-f}
G = f(y_1x_1 y_2^{-1})x_1 y_2,
\end{equation}
where $f$ is an arbitrary function/series of/in a single variable.
In that case the resulting fields  $[(\bar{x}_1,\bar{y}_1), (\bar{x}_2,\bar{y}_2)]$ satisfy not only the reduction conditions \eqref{eq:u-l}, but also the Weyl commutation relations
\eqref{eq:Weyl-c-r} with the same constant $q$.	
\end{Prop}
\begin{proof}
	The commutation of $\bar{x}_2 \bar{y}_2^{-1}$ with $\bar{x}_1$ leads to the condition
	\begin{equation}
	y_1^{-1} x_1^{-1} y_1 x_1 = x_2 y_2^{-1} x_2^{-1} y_2 ,
	\end{equation}
fulfilled due to \eqref{eq:Weyl-c-r}. Then the commutation of $\bar{y}_1$ with $\bar{x}_2$ and $\bar{y}_2$ gives, correspondingly, 
\begin{equation} \label{eq:G-1}
G y_1 x_2 = y_1 x_2 G \qquad \text{and} \quad G y_2 = y_2 G.
\end{equation}
Finally, the commutation of of $\bar{x}_1$ with $\bar{x}_2$ leads to
\begin{equation}\label{eq:G-2}
G x_1 x_2 y_2 = y_2 x_1 x_2 G \qquad \text{or} \quad q G x_1 x_2 = x_1 x_2 G .
\end{equation}
Conditions \eqref{eq:G-1} and \eqref{eq:G-2} are satisfied by $G$ of the form \eqref{eq:G-f}, and imply the Weyl commutation relations for the transformed fields.
\end{proof}
\begin{Cor}
	For commuting variables the homogeneous Veblen map \eqref{eq:V-h} with the gauge function given by ~\eqref{eq:G-f} is a Poisson map with the brackets \eqref{eq:Poisson-Weyl}.
\end{Cor}

\begin{Ex} \label{ex:G-ab}
	In the simplest case $f(z) = \alpha z + \beta $, where $\alpha$ and $\beta$ are central parameters (which we assume do not vanish simultaneously), the  ultra-local/quantum/Poisson Veblen map   $\mathcal{V}^{(\alpha,\beta)}$ is given by
\begin{equation} \label{eq:V-h-ab}
\begin{array}{ll}
\bar{x}_1 = & x_1 x_2, \\
\bar{x}_2 = & y_1 x_2 (\alpha y_1 + \beta x_1 y_2 )^{-1},
\end{array} \qquad
\begin{array}{ll}
\bar{y}_1 = &\alpha y_1 + \beta x_1 y_2 ,\\
\bar{y}_2 = & y_2 (\alpha y_1 + \beta x_1 y_2)^{-1},
\end{array}
\end{equation}
and its inverse map reads
\begin{equation} \label{eq:V-h-inv-ab}
\begin{array}{ll}
x_1 = & \bar{x}_1 (\alpha \bar{x}_2 + \beta \bar{x}_1 \bar{y}_2  )^{-1}, \\
x_2 = & \alpha \bar{x}_2 + \beta \bar{x}_1 \bar{y}_2 ,
\end{array} \qquad
\begin{array}{ll}
y_1 = &  \bar{y}_1 \bar{x}_2 (\alpha \bar{x}_2 + \beta \bar{x}_1 \bar{y}_2 )^{-1},  \\
y_2 = & \bar{y}_1 \bar{y}_2  .
\end{array}
\end{equation}	
\end{Ex}

The corresponding analogue of Proposition~\ref{prop:V-p-h} shows how the central parameters of the ultra-local/quantum/Poisson Veblen map~\eqref{eq:V-h-ab} should be matched in the pentagonal equation.
\begin{Prop} \label{prop:V-p-h-ab}
	The homogeneous ultra-local/quantum/Poisson Veblen map \eqref{eq:V-h-ab}
	satisfies the reversed pentagonal condition
	\begin{equation} 
	\mathcal{V}_{23}^{(\alpha_C,\beta_C)} \circ \mathcal{V}_{13}^{(\alpha_B,\beta_B)} \circ  \mathcal{V}_{12}^{(\alpha_A,\beta_A)} = \mathcal{V}_{12}^{(\alpha_H,\beta_H)} \circ  \mathcal{V}_{23}^{(\alpha_G,\beta_G)} ,
	\end{equation}
	provided the central parameters of the maps satisfy conditions
	\begin{equation} \label{eq:V-pent-h-ab}
	\alpha_G = \alpha_B \alpha_C, \qquad  \alpha_H = \alpha_A \alpha_B, \qquad
	\beta_A = \alpha_C \beta_H, \qquad \beta_B = \beta_G \beta_H, \qquad \beta_C = \alpha_A \beta_G	.
	\end{equation}
\end{Prop}
\begin{proof}
Let us denote
\begin{align*}
[(\bar{x}_1 , \bar{y}_1) , (\bar{x}_2, \bar{y}_2) , (\bar{x}_3 , \bar{y}_3 )] = &
\mathcal{V}_{23}^G [(x_1,y_1),(x_2,y_2), (x_3,y_3)] \\
[(x_1^\prime , y_1^\prime) , (x_2^\prime, y_2^\prime) , (x_3^\prime , y_3^\prime )] = &
 \mathcal{V}_{12}^A[(x_1,y_1),(x_2,y_2), (x_3,y_3)] . 
\end{align*}
The first equation $H=B$ of \eqref{eq:V-pent-g-ab} reads then
\begin{equation*}
\alpha_H \bar{y}_1 + \beta_H \bar{x}_1 \bar{y}_2 = \alpha_B y^\prime_1 + \beta_B x^\prime_1  y^\prime_3
\end{equation*}	
which gives
\begin{equation}
\alpha_H y_1 + \beta_H x_1 (\alpha_G y_2 + \beta_G x_2 x_3) = \alpha_B ( \alpha_A y_1 + \beta_A x_1 y_2) + \beta_B x_1 x_2  y_3.
\end{equation}	
By comparing coefficients of monomials on both sides we obtain
\begin{equation} \label{eq:H=B-ab}
\alpha_H = \alpha_A \alpha_B, \qquad \beta_B = \beta_G \beta_H, \qquad
\alpha_B \beta_A = \alpha_G \beta_H.
\end{equation}
Application of the same reasoning for the second equation $G=CB$ of \eqref{eq:V-pent-g-ab} gives, 
\begin{equation} \label{eq:G=BC-ab}
\alpha_G = \alpha_B \alpha_C, \qquad \beta_C = \alpha_A \beta_G, \qquad \alpha_C \beta_B = \beta_A \beta_G.
\end{equation}
Inserting the second expression of equation~\eqref{eq:H=B-ab} into the third expression of equation~\eqref{eq:G=BC-ab} gives the third one of the system~\eqref{eq:V-pent-h-ab}, thus completing the list. The last four-term expression of~\eqref{eq:H=B-ab} is then a consequence of~\eqref{eq:V-pent-h-ab}.  
\end{proof}

\section{The homogeneous Hirota map and its quantum reduction} \label{sec:Zam-h}
\subsection{The homogeneous Hirota map and local Yang--Baxter equations} \label{sec:Hir-h}
Let us start with generic homogeneous version of the linear problem, studied before in the Hirota gauge~\eqref{eq:lin-H} and in the affine version~\eqref{eq:lin-aff}
\begin{equation} \label{eq:lin-Z-h}
\begin{array}{lll}
\bpsi_{(2)}& =&\bpsi x_1 + \bpsi_{(1)} y_1 \\
\bpsi_{(23)}& =& \bpsi_{(2)} x_2 + \bpsi_{(12)} y_2 \\
\bpsi_{(3)} &= &\bpsi  x_3 + \bpsi_{(2)} y_3
\end{array}
\qquad \text{and} \qquad
\begin{array}{lll}
\bpsi_{(23)}  &= &\bpsi_{(3)}\tilde{x}_1  + \bpsi_{(13)} \tilde{y}_1 , \\
\bpsi_{(3)} &= &\bpsi \tilde{x}_2 + \bpsi_{(1)} \tilde{y}_2 , \\
\bpsi_{(13)}  & = &\bpsi_{(1)} \tilde{x}_3 + \bpsi_{(12)} \tilde{y}_3 .
\end{array}
\end{equation}
The transformation formulas read
\begin{equation}
\label{eq:Z-h}
\begin{array}{ll}
\tilde{x}_1 = & (x_3 + x_1 y_3)^{-1} x_1 x_2 \\
\tilde{x}_2 =  & x_3 + x_1 y_3 \\
\tilde{x}_3 = & y_1 x_1^{-1} x_3 (x_3 + x_1 y_3)^{-1} x_1 x_2 G^{-1}
\end{array} \qquad 
\begin{array}{ll}
\tilde{y}_1 = & G , \\
\tilde{y}_2 =  & y_1 y_3 , \\
\tilde{y}_3 = & y_2 G^{-1} ,
\end{array}
\end{equation}
where the free gauge parameter $G=\tilde{y}_1$ reflects the multiplicative freedom  in definition of the homogeneous coordinates $\bpsi_{(13)} $. 

Directly one can show that the transformation 
\begin{equation} \label{eq:F--->}
\mathcal{F}^G \colon [(x_1,y_1),(x_2,y_2),(x_3,y_3)] \dashrightarrow [(\tilde{x}_1,\tilde{y}_1), (\tilde{x}_2,\tilde{y}_2), (\tilde{x}_3,\tilde{y}_3)]
\end{equation}
can be decomposed into the homogeneous version of~\eqref{eq:fact-Z}
\begin{equation} \label{eq:fact-Z-h}
	\mathcal{F}^G = P_{23} \circ \mathcal{V}^G_{12} \circ \mathcal{N}_{13},
	\end{equation}
where $\mathcal{V}^G$ and $\mathcal{N}$ are given by~\eqref{eq:V-h} and \eqref{eq:N-h}, respectively. 
\begin{Cor}
	The homogeneous version of the Hirota map allows for the affine/non-homogeneous reduction~\eqref{eq:F-map}.
\end{Cor}
Let us describe the geometric interpretation of the Hirota map and of its linear problem (in all the gauges considered in the paper) within the local Yang-Baxter  equation approach~\cite{MailletNijhoff,Kashaev-LMP,Sergeev-LMP,Kashaev-Korepanov-Sergeev}. It is known that such a description of the mapping provides a simple proof of its Zamolodchikov property. However, to analyse generic case of arbitrary $2\times 2$ matrix $L$ one is forced to consider the map up to certain gauge transformations~\cite{Korepanov-SI,Kashaev-Korepanov-Sergeev}. We will describe in detail such an approach in the case of the homogeneous Hirota map.

Consider $2\times 2$ matrix $L(x,y)=\left( \begin{array}{cc} x & 1 \\ y & 0 \end{array} \right)$ and define its $3\times 3$ extensions by
\begin{equation*}
L_{12}(x,y) = \left( \begin{array}{ccc} x & 1 & 0 \\ y & 0 & 0 \\ 0 & 0 & 1\end{array} \right) , \quad
L_{13}(x,y) = \left( \begin{array}{ccc} x & 0 & 1 \\ 0 & 1 & 0 \\ y & 0 & 0\end{array} \right) , \quad
L_{23}(x,y) = \left( \begin{array}{ccc} 1 & 0 & 0 \\ 0 & x & 1 \\ 0 & y & 0\end{array} \right) .
\end{equation*}
Multiplication from the right of the system of vectors 
$(\bpsi, \bpsi_{(1)}, \bpsi_{(12)})$ by the matrix $L_{12}(x_1,y_1)$, in view of the first equation of the linear system \eqref{eq:lin-Z-h}, gives
\begin{equation*}
(\bpsi, \bpsi_{(1)}, \bpsi_{(12)})\left( \begin{array}{ccc} x_1 & 1 & 0 \\ y_1 & 0 & 0 \\ 0 & 0 & 1\end{array} \right) = (\bpsi_{(2)}, \bpsi, \bpsi_{(12)}).
\end{equation*} 
The first part of the linear problem \eqref{eq:lin-Z-h} provides therefore transformations
\begin{equation*}
(\bpsi, \bpsi_{(1)}, \bpsi_{(12)}) \xrightarrow{L_{12}(x_1,y_1)}
(\bpsi_{(2)}, \bpsi, \bpsi_{(12)}) \xrightarrow{L_{13}(x_2,y_2)}
(\bpsi_{(23)}, \bpsi, \bpsi_{(2)}) \xrightarrow{L_{23}(x_3,y_3)}
(\bpsi_{(23)}, \bpsi_{(3)}, \bpsi) ,
\end{equation*}
while its second part gives 
\begin{equation*}
(\bpsi, \bpsi_{(1)}, \bpsi_{(12)}) \xrightarrow{L_{23}(\tilde{x}_3,\tilde{y}_3)}
(\bpsi, \bpsi_{(13)}, \bpsi_{(1)}) \xrightarrow{L_{13}(\tilde{x}_2,\tilde{y}_2))}
(\bpsi_{(3)}, \bpsi_{(13)}, \bpsi) \xrightarrow{L_{12}(\tilde{x}_1,\tilde{y}_1)}
(\bpsi_{(23)}, \bpsi_{(3)}, \bpsi) .
\end{equation*}
The local Yang--Baxter equation 
\begin{equation} \label{eq:lYB}
L_{12}(x_1,y_1) L_{13}(x_2,y_2) L_{23}(x_3,y_3) = 
L_{23}(\tilde{x}_3,\tilde{y}_3) L_{13}(\tilde{x}_2,\tilde{y}_2) L_{12}(\tilde{x}_1,\tilde{y}_1),
\end{equation}
which is another form of the compatibility condition of the linear problem \eqref{eq:lin-Z-h}, reads explicitly as follows
\begin{equation*}
\left( \begin{array}{ccc} x_1 x_2 & x_3 + x_1 y_3 & 1 \\ y_1 x_2 & y_1 y_3 & 0 \\ 
y_2 & 0 & 0\end{array} \right) = 
\left( \begin{array}{ccc} \tilde{x}_2 \tilde{x}_1 & \tilde{x}_2 & 1 \\ 
\tilde{y}_2 \tilde{x}_1 + \tilde{x}_3 \tilde{y}_1 & \tilde{y}_2 & 0 \\ 
\tilde{y}_3 \tilde{y}_1 & 0 & 0\end{array} \right).
\end{equation*} 
Interpreted as five equations for six unknowns, the above system can be solved parametrizing first $\tilde{y}_1 = G$, and gives then the transformation formulas~\eqref{eq:Z-h}.

In the affine reduction $y_i = 1 - x_i$, $i=1,2,3$, it can be checked that by putting $\tilde{y}_1 = G = 1 - \tilde{x}_1$ we preserve the reduction condition also in the other variables. The corresponding matrix $L(x,1-x) = L^A(x)$ when inserted in the local Yang--Baxter equation gives rise to transformation formulas~\eqref{eq:F-map}.

By comparing the general linear problem in homogeneous coordinates~\eqref{eq:lin-Z-h} with the linear problem in the Hirota gauge~\eqref{eq:lin-H} one can obtain the corresponding reduction 
\begin{equation*}
L(-y,1) = \left( \begin{array}{cc} -y & 1 \\ 1 & 0 \end{array} \right)
\end{equation*} 
of the matrix entering the local Yang--Baxter equation~\eqref{eq:lYB}. When inserted into the equation the matrix gives the Hirota map~\eqref{eq:Hirota-map}. Up to the sign change symmetry it agrees with the matrix $L^H(y)$ of Section~\ref{sec:Hir-Zam-H}.

\subsection{The Zamolodchikov property of the homogeneous Hirota map}
We start from the ten-term relation for the normalization and Veblen maps in homogeneous coordinates.
\begin{Prop} \label{prop:NV-tt-h}
	The normalization map \eqref{eq:N-h} and the Veblen map \eqref{eq:V-h} are paired by the homogeneous form of the ten-term equation 
	\begin{equation} \label{eq:ten-term-h}
	\mathcal{V}_{13}^C\circ \mathcal{N}_{12}\circ \mathcal{V}_{14}^B\circ \mathcal{N}_{34}\circ \mathcal{V}_{24}^A = \mathcal{N}_{34}\circ \mathcal{V}_{24}^H \circ \mathcal{N}_{14}\circ \mathcal{V}_{13}^G\circ \mathcal{N}_{12}  ,
	\end{equation}
provided the gauge parameters of the Veblen transformations satisfy conditions
	\begin{equation} \label{eq:tt-g}
H = BA, \qquad G = C \left[ 1 + y_4 x_4^{-1} (x_2 + x_1 y_2)^{-1} x_1 x_3 \right].
\end{equation}	
\end{Prop}
\begin{proof}
	We present just the final expressions
	\begin{equation*}
	[(x_1^\prime , y_1^\prime) , (x_2^\prime, y_2^\prime) , (x_3^\prime , y_3^\prime ), (x_4^\prime , y_4^\prime )] = 
	\begin{cases}
	\mathcal{V}_{13}^C\circ \mathcal{N}_{12}\circ \mathcal{V}_{14}^B\circ \mathcal{N}_{34}\circ \mathcal{V}_{24}^A [(x_1,y_1),(x_2,y_2), (x_3,y_3),(x_4,y_4)], \\ 
	\mathcal{N}_{34}\circ \mathcal{V}_{24}^H \circ \mathcal{N}_{14}\circ \mathcal{V}_{13}^G\circ \mathcal{N}_{12}[(x_1,y_1),(x_2,y_2), (x_3,y_3),(x_4,y_4)] ,
	\end{cases}
	\end{equation*}
	where
	\begin{align*}
	x_1^\prime & = (x_2 x_4 + x_1 y_2 x_4 + x_1 x_3 y_4)^{-1}x_1 x_3 , \\
	y_1^\prime & = C = G \left[ 1 + y_4 x_4^{-1} (x_2 + x_1 y_2)^{-1} x_1 x_3 \right]^{-1} ,\\
	x_2^\prime & = x_2 x_4 + x_1 y_2 x_4 + x_1 x_3 y_4 ,\\
	y_2^\prime & = BA = H, \\
	x_3^\prime & = BA (y_2x_4 + x_3 y_4)^{-1} x_1^{-1} x_2 x_4 \left[ (x_2 + x_1 y_2)^{-1} x_4 + x_1 x_3 y_4 \right]^{-1} x_1 x_3 C^{-1} = \\
	  & = H  (y_2x_4 + x_3 y_4)^{-1} x_1^{-1} x_2  (x_2 + x_1 y_2)^{-1} x_1 x_3 G^{-1}, \\
	y_3^\prime & = y_3\left( 1 + y_4 x_4^{-1}y_2^{-1} x_3 \right)^{-1} C^{-1}
	= y_3\left( 1 + y_4 x_4^{-1}y_2^{-1} x_3 \right)^{-1}  
	\left[ 1 + y_4 x_4^{-1} (x_2 + x_1 y_2)^{-1} x_1 x_3 \right] G^{-1} , \\
	x_4^\prime & = y_1(y_2x_4 + x_3 y_4)A^{-1} B^{-1} = y_1(y_2x_4 + x_3 y_4)H^{-1} , \\
	y_4^\prime & = x_3 y_4 A^{-1} B^{-1} =  x_3 y_4 H^{-1} .
	\end{align*}
\end{proof}

Let us state the final result for the Hirota map in arbitrary homogeneous coordinates.
\begin{Prop}
	The transformation $\mathcal{F}^G$ given by equations~\eqref{eq:Z-h} satisfies the Zamolodchikov condition
	\begin{equation*}
	\mathcal{F}_{123}^D \circ \mathcal{F}_{145}^C \circ \mathcal{F}_{246}^B\circ  \mathcal{F}_{356}^A = \mathcal{F}_{356}^V \circ \mathcal{F}_{246}^U \circ \mathcal{F}_{145}^T \circ \mathcal{F}_{123}^S,
	\end{equation*}
	provided the gauge parameters satisfy equations
	\begin{equation}
	B=VT, \qquad D=T, \qquad CA = U.
	\end{equation}
\end{Prop}
\begin{proof}
	Following the proof of Lemma~\ref{lem:tt} we arrive to the decomposition
	\begin{equation*}
	\mathcal{V}_{13}^D\circ\mathcal{V}_{36}^B\circ
	\mathcal{N}_{12}\circ \mathcal{V}_{15}^C\circ \mathcal{N}_{35}\circ \mathcal{V}_{25}^A 
	\circ \mathcal{N}_{14}\circ \mathcal{N}_{24}
	= \mathcal{V}_{36}^V\circ\mathcal{V}_{16}^T\circ
	\mathcal{N}_{35}\circ \mathcal{V}_{25}^U \circ \mathcal{N}_{15}\circ \mathcal{V}_{13}^S\circ \mathcal{N}_{24}  \circ \mathcal{N}_{12} .
	\end{equation*}	
	Because of the pentagon relations
	\begin{equation*}
	\mathcal{N}_{24}  \circ \mathcal{N}_{12} = \mathcal{N}_{12}  \circ \mathcal{N}_{14}  \circ \mathcal{N}_{24} \qquad \text{and} \qquad
	\mathcal{V}_{13}^D\circ\mathcal{V}_{36}^B = \mathcal{V}_{36}^V\circ\mathcal{V}_{16}^T \circ \mathcal{V}_{13}^X
	\end{equation*}
	where, by Proposition~\ref{prop:V-p-h}, $D=T$ and $B=VT$ we obtain
	\begin{equation*}
	\mathcal{V}_{13}^X\circ
	\mathcal{N}_{12}\circ \mathcal{V}_{15}^C\circ \mathcal{N}_{35}\circ \mathcal{V}_{25}^A 
	= 
	\mathcal{N}_{35}\circ \mathcal{V}_{25}^U \circ \mathcal{N}_{15}\circ \mathcal{V}_{13}^S\circ \mathcal{N}_{12} .
	\end{equation*}
	Finally, by Proposition~\ref{prop:NV-tt-h} the remaining relations between the gauge parameters read
	\begin{equation*}
	CA = U, \qquad S = X\left[ 1 + y_5 x_5^{-1} (x_2 + x_1 y_2)^{-1} x_1 x_3 \right].
	\end{equation*}
\end{proof}
\begin{Rem}
	The connection between $S$ and $X$ is needed when the decomposition of the Zamolodchikov equation into the ten-term relation and pentagon equations is used.
\end{Rem}

\subsection{Quantum and Poisson reductions of the homogeneous Hirota map} \label{sec:Hir-q-P}
Let us present the corresponding analogue of Proposition~\ref{prop:V-p-h-ab} which shows how the central parameters of the ultra-local/quantum/Poisson Veblen map~\eqref{eq:V-h-ab} should be matched in the ten-term relation.
\begin{Prop} \label{prop:tt-ab}
	The ultra-local/quantum/Poisson reductions of the homogeneous normalization and Veblen maps are paired by the corresponding form of the ten-term relation
\begin{equation} \label{eq:ten-term-q}
\mathcal{V}_{13}^{(\alpha_C,\beta_C)}\circ \mathcal{N}_{12}\circ \mathcal{V}_{14}^{(\alpha_B,\beta_B)}\circ \mathcal{N}_{34}\circ \mathcal{V}_{24}^{(\alpha_A,\beta_A)} = \mathcal{N}_{34}\circ \mathcal{V}_{24}^{(\alpha_H,\beta_H)} \circ \mathcal{N}_{14}\circ \mathcal{V}_{13}^{(\alpha_G,\beta_G)}\circ \mathcal{N}_{12}  ,
\end{equation}	
provided the central parameters of the ultra-local/quantum/Poisson Veblen map satisfy conditions
\begin{equation} \label{eq:tt-ab}
\alpha_G = \alpha_B \alpha_C, \quad
\alpha_H = \alpha_A \alpha_B, \quad 
\beta_A = \alpha_C \beta_H, \quad
\beta_B = \beta_G \beta_H, \quad 
\beta_C = \alpha_A \beta_G.
\end{equation}	
\end{Prop}
\begin{proof}
	The technique of the demonstration of the above relations is the same like in the proof of Proposition~\ref{prop:V-p-h-ab}; notice that also equations~\eqref{eq:V-pent-h-ab}
and \eqref{eq:tt-ab} are identical. From the first relation of~\eqref{eq:tt-g} we get equations~\eqref{eq:H=B-ab}, while the second one gives equations~\eqref{eq:G=BC-ab}.
\end{proof}

In the rest of this Section we discuss the central parameters describing the gauge in Example~\ref{ex:G-ab} and Proposition~\ref{prop:V-p-h-ab}. In particular we show how the parameters should be combined in the corresponding homogeneous Hirota maps in order the Zamolodchikov equation be satisfied. Directly by properties of the ultra-local/quantum/Poisson reductions of the homogeneous normalization and Veblen maps we obtain the following result.
\begin{Prop} \label{prop:F-q}
The map $\mathcal{F}^{(\alpha,\beta)} = P_{23} \circ \mathcal{V}_{12}^{(\alpha,\beta)} \circ \mathcal{N}_{13}$ given explicitly by
\begin{equation}
\label{eq:Z-h-q}
\begin{array}{ll}
\tilde{x}_1 = & x_1 x_2 (x_3 + x_1 y_3)^{-1},  \\
\tilde{x}_2 =  & x_3 + x_1 y_3 ,\\
\tilde{x}_3 = & y_1 x_2  x_3 (\alpha y_1 x_3 + \beta x_1 y_2)^{-1} ,
\end{array} \qquad 
\begin{array}{ll}
\tilde{y}_1 = & (\alpha y_1 x_3 + \beta x_1 y_2)(x_3 + x_1 y_3)^{-1},  \\
\tilde{y}_2 =  & y_1 y_3 , \\
\tilde{y}_3 = & y_2 (x_3 + x_1 y_3) (\alpha y_1 x_3 + \beta x_1 y_2)^{-1} ,
\end{array}
\end{equation}
preserves the Weyl commutation relations
\begin{equation}
x_i y_i = q y_i x_i \quad \text{and} \quad x_i x_j = x_j x_i, \quad y_i y_j = y_j y_i, \quad x_i y_j = y_j x_i, \quad \text{for} \quad i\neq j,
\end{equation}
or, for commuting variables, the corresponding Poisson brackets
\begin{equation}
\{ x_i , y_i \} = x_i y_i \quad \text{and} \quad \{x_i ,x_j \} = \{ y_i , y_j \} = \{ x_i,  y_j\} = 0 \quad \text{for} \quad i\neq j.
\end{equation}
\end{Prop}

\begin{Prop} \label{prop:F-h-ab-Z}
The ultra-local/quantum/Poisson Hirota map satisfies the functional Zamolodchikov equation
\begin{equation} \label{eq:F-ab}
\mathcal{F}_{123}^{(\alpha_D,\beta_D)} \circ \mathcal{F}_{145}^{(\alpha_C,\beta_C)} \circ \mathcal{F}_{246}^{(\alpha_B,\beta_B)}\circ  \mathcal{F}_{356}^{(\alpha_A,\beta_A)} = \mathcal{F}_{356}^{(\alpha_V,\beta_V)} \circ \mathcal{F}_{246}^{(\alpha_U,\beta_U)} \circ \mathcal{F}_{145}^{(\alpha_T,\beta_T)} \circ \mathcal{F}_{123}^{(\alpha_S,\beta_S)},
\end{equation}
provided the central parameters satisfy conditions
\begin{gather*}
\alpha_B = \alpha_T \alpha_V, \quad \alpha_U = \alpha_A \alpha_C, \quad \beta_C = \beta_S \beta_U, \quad \beta_T = \beta_B \beta_D , \\
\alpha_A \beta_S = \alpha_V \beta_D, \quad \alpha_C \alpha_D = \alpha_S \alpha_T, \quad
\alpha_C \beta_A = \alpha_S \beta_U, \quad \beta_A \beta_B = \beta_U \beta_V.
\end{gather*}
\end{Prop}
\begin{proof}
	By comparison of the indices of the Veblen part of \eqref{eq:F-ab} with those in Proposition~\ref{prop:V-p-h-ab} we obtain relations
	\begin{equation*} 
	\alpha_B = \alpha_T \alpha_V, \quad
	\alpha_D = \alpha_X \alpha_T, \quad 
	\beta_X = \alpha_V \beta_D, \quad
	\beta_T = \beta_B \beta_D, \quad 
	\beta_V = \alpha_X \beta_B,
	\end{equation*}	
	while comparison of the ten-term part of \eqref{eq:F-ab} with those of Proposition~\ref{prop:tt-ab} gives
	\begin{equation*} 
	\alpha_S = \alpha_C \alpha_X, \quad
	\alpha_U = \alpha_A \alpha_C, \quad 
	\beta_A = \alpha_X \beta_U, \quad
	\beta_C = \beta_S \beta_U, \quad 
	\beta_X = \alpha_A \beta_S.
	\end{equation*}	
	The statement is the result of elimination of $\alpha_X$ and $\beta_X$ from the above two systems.
\end{proof}
\begin{Rem}
	As the set of independent eight central parameters one can take:
	(i) $\alpha_C, \alpha_T, \beta_B, \beta_U$, (ii)~three out of $\alpha_A, \alpha_V, \beta_D, \beta_S$, (iii) one out of $\alpha_D, \alpha_S, \beta_A, \beta_V$.
\end{Rem}

\section{Conclusion}
We constructed several solutions of the functional Zamolodchikov equation in terms of rational functions of totally non-commutative variables (of course the Zamolodchikov equation holds also in the commutative reduction of the maps). Our examples come from the geometric description of the non-Abelian Hirota--Miwa equation~\cite{Nimmo-NCKP} in terms of Desargues lattices~\cite{Dol-Des} --- the same geometric statements can be algebraically described in terms of different maps. The most general case of homogeneous coordinates gives, by the ultra-locality assumption, the reduction preserving the Weyl commutation relations between the fields and, in the classical limit, the corresponding Poisson map. 

In our approach we used previously found link~\cite{DoliwaSergeev-pentagon} between the Desargues lattices (as described in homogeneous coordinates) and the normalization and Veblen pentagonal maps. In particular we gave the geometric meaning to the corresponding ten-term relation~\cite{Kashaev-Sergeev} needed to produce Zamolodchikov map from pentagonal maps. Notice that our quantum maps have been obtained by application of the ultra-locality constraint on non-commutative maps (i.e. maps given in terms of fully non-commutative symbols), see also \cite{Sergeev-Q2+1} for an earlier work in this context. This way is in opposite direction to the standard quantization procedure which starts from Poisson maps. It is well known that in transition from maps defined in terms of commuting variables to their quantum versions different ordering of factors gives different theories. Even worse, some factors, which are crucial to integrability, may cancel out in the commutative reduction (see equation~\eqref{eq:N-h} as an example). By geometric identification of integrability (here the Desargues theorem) we obtain for free \emph{integrable} non-commutative mappings from the commutative ones.

It is known  that  Hirota's discrete KP system and its various reductions give rise to majority of integrable equations and appear ``everywhere'' in the theory of solvable systems of statistical and quantum physics~\cite{KNS-rev}. It would be desirable to see (if possible) how the corresponding reductions of the solutions of the Zamolodchikov equation presented here descend to the level of particular models.

\subsection*{Acknowledgements} The research of A.D. was supported by National Science Centre, Poland, under grant 2015/19/B/ST2/03575 \emph{Discrete integrable systems -- theory and applications}.
\bibliographystyle{amsplain}
\providecommand{\bysame}{\leavevmode\hbox to3em{\hrulefill}\thinspace}

\end{document}